\documentclass[10pt]{article}

\usepackage{amsmath, latexsym, amssymb, amscd,graphicx,amsthm,amstext,amsxtra,amsopn,amsbsy,young,enumitem}
\usepackage[vcentermath]{youngtab}

\newtheorem{theorem}{\bf Theorem}[section]
\newtheorem{ex}[theorem]{\bf Example}

\newtheorem{definition}[theorem]{\textsl{\bf Definition}{}}
\newtheorem{lemma}[theorem]{\bf Lemma}
\newtheorem{corollary}[theorem]{\bf Corollary}

\newtheorem{remark}[theorem]{\bf Remark}

\numberwithin{equation}{section}

\begin{document}

\title{Quantum U-statistics}

\author{ M\u{a}d\u{a}lin Gu\c{t}\u{a}$^{1}$ and Cristina Butucea$^{2}$ \\\\
$^{1}$ University of Nottingham, \\School of Mathematical Sciences, \\University Park, NG7 2RD, Nottingham, U.K. \\\\
$^{2}$ Laboratoire Paul Painlev\'{e} (UMR CNRS 8524), \\
Universit\'{e} des Sciences et Technologies de Lille 1, \\
59655 Villeneuve dÕAscq cedex, France 
}

\date{}

\maketitle
\begin{abstract}
The notion of a $U$-statistic for an $n$-tuple of identical quantum systems is introduced in analogy to the classical (commutative) case: given a selfadjoint `kernel' $K$ acting on $(\mathbb{C}^{d})^{\otimes r}$ with $r<n$, we define the symmetric operator $U_{n}= {n \choose r} \sum_{\beta}K^{(\beta)}$ with 
$K^{(\beta)}$ being the kernel acting on the subset $\beta$ of $\{1,\dots ,n\}$. 
If the systems are prepared in the i.i.d state $\rho^{\otimes n}$ it is shown that the sequence of properly normalised $U$-statistics converges in moments to a linear combination of Hermite polynomials in canonical variables of a CCR algebra defined through the Quantum Central Limit Theorem. In the special cases of non-degenerate kernels and kernels of order $2$  it is shown that the convergence holds in the stronger distribution sense.

Two types of applications in quantum statistics are described: 
testing beyond the two simple hypotheses scenario, and quantum metrology with interacting hamiltonians. 

\end{abstract}


\section{Introduction}

Let $X_{1},\dots, X_{n}$ be an i.i.d. sample from a distribution $\mathbb{P}$ over 
$\mathbb{R}$. 
A $U$-statistic is  a generalisation of the notion of sample mean
$$
\bar{X}= \frac{1}{n}\sum_{i=1}^{n} X_{i}
$$
to the case where one averages a (symmetric) kernel $h(x_{1},\dots ,x_{r})$ over all possible ways of choose $r$ different variables out of $(X_{1},\dots,X_{n})$: 
$$
U_{n}= {n \choose r}^{-1} \sum_{\beta} h(X_{\beta_{1}} , \dots ,X_{\beta_{r}}).
$$   
The study of $U$-statistics is a classic topic in mathematical statistics \cite{Serfling,vanderVaart}, whose roots go back to the papers by Hoeffding \cite{Hoeffding} and Halmos \cite{Halmos}. Numerous examples of such statistics are used in both estimation and testing, their main advantages being the intuitive character and tractable asymptotic properties.

In this work we introduce the notion of $U$-statistic for non-commutative variables appearing in quantum mechanics, prove an analogue of the `classical' convergence theorem, and sketch a few applications in quantum statistics. We take the view that Quantum Statistics is an extension of classical (commutative) statistics and that techniques from the latter can be generalised and used to solve quantum statistical problems. This approach proved its fruitfulness right from the beginning of the field \cite{Helstrom,Holevo,Belavkin} with the development of quantum Cram\'{e}r-Rao inequalities and quantum Fisher information techniques, the analysis of group covariant models, and more recently the (asymptotic) solutions of the optimal state discrimination problem \cite{Hiai&Petz,Ogawa&Nagaoka,Audenaert&Szkola}, the optimal state estimation problem \cite{Bagan&Gill,Hayashi&Matsumoto,Guta&Janssens&Kahn,Guta&Kahn2}, 
the theory of quantum sufficiency \cite{Petz&Jencova}.

Our study of quantum $U$-statistics is motivated by their applications in estimation and testing for large samples of identically prepared, independent 
quantum systems. Consider $n$ such systems prepare in the unknown state 
$\rho\in M(\mathbb{C}^{d})$ and suppose that we are interested in a certain (non-linear) functional of $\theta=f(\rho)$ which has the form
$$
\theta = {\rm Tr}(\rho^{\otimes r} K)
$$ 
where $K$ is a selfadjoint operator on $(\mathbb{C}^{d})^{\otimes r}$ which can be chosen to be symmetric under permutations. Then by measuring $K$ we obtain an unbiased estimator of $\theta$, but in this way we use only $r$ of the $n$ systems. The linear combination of all such kernels $K^{(\beta)}$ acting on different tuples of systems has the same basic property but has the advantage of a smaller variance
\begin{equation}\label{eq.1}
U_{n}= {n \choose r}^{-1} \sum_{\beta} K^{(\beta)}.
\end{equation}

Another application is parameter estimation in quantum metrology \cite{Boixo,Boixo2,Roy} . Assume that $H_{n}:= {n \choose r}U_{n}$ is a multi-body hamiltonian with $r$-body interactions, and let $\psi^{\gamma,n}_{t}:= 
\exp(i t \gamma H_{n})\psi^{\otimes n}$ be the evolved state at time $t$ starting from a product state. In \cite{Boixo2} it is shown that $\gamma$ can be estimated with precision $n^{-r+1/2}$ by means of separate measurements on the $n$ systems. Already for $r=2$ this rate is faster than the famous `Heisenberg limit' $n^{-1}$ which is the absolute estimation rate for non-interacting hamiltonians!

Although definition \eqref{eq.1} is identical to that of classical $U$-statistics, the existing theory cannot be applied since in general the terms of the sum do not commute with each other. However, guided by the classical technique called the 
{\it Hoeffding decomposition} we show that $U_{n}$ can be written as a 
polynomial in {\it fluctuation observables} 
$\mathbb{F}_{n}(A):= \frac{1}{\sqrt{n}}\sum_{i=1}^{n} A^{(i)}$ and 
{\it empirical averages} $\mathbb{P}_{n}(B):= \frac{1}{n}\sum_{i=1}^{n} B^{(i)}$.  By applying a version of the Quantum Central Limit Theorem \cite{Petz} we obtain the main result, the convergence in moments Theorem \ref{th.conv.moments.ustat}. This says that asymptotically with $n$, the moments of the rescaled $U$-statistic $n^{c/2}U_{n}$ converge to moments of a certain polynomial in quantum canonical variables of a CCR algebra with respect to a Gaussian state. The integer $1\leq c\leq r$ is a degree of degeneracy and depends on the statistics of $K$ with respect to the state $\rho^{\otimes r}$. 
In important special cases, such as order 2 kernels and non-degenerate kernels this can be strengthened to convergence in distribution. In particular, in the latter case, the asymptotic distribution is normal, similar to the Central Limit Theorem 
for fluctuations of collective observables.

The paper is organised as follows. In section \ref{sec.classical} we give a short review of the classical theory of $U$-statistics culminating with the convergence Theorem \ref{th.convergence.classic}. Section \ref{sec.oscillator} deals with standard notions on Gaussian states of harmonic oscillators and ends with Theorem \ref{th.orthog.hermite.noncomm} which will be needed in the main result. The theorem shows that the symmetric products of Hermite polynomials in the canonical variables are orthogonal on the lower order polynomials, with respect to the inner product defined by a thermal equilibrium state. Section \ref{sec.clt} formulates the Quantum Central Limit Theorem \cite{Petz} for $d$ dimensional systems and gives an explicit description of the canonical commutation relations (CCR) algebra emerging in the limit, as a product of $d(d-1)/2$ harmonic oscillators in thermal states and a $(d-1)$-dimensional classical Gaussian variables.

In section \ref{sec.quantum.u.stat} we introduce the notion of quantum 
$U$-statistic and analyse its asymptotic behaviour by using the Hoeffding decomposition for $L^{2}$-spaces of non-commuting variables derived in section \ref{sec.q.hoeffding}. The main result of the paper is 
Theorem \ref{th.conv.moments.ustat} which shows that properly normalised $U$-statistics converge in moments to a certain polynomial in the 
canonical variables of the CCR algebra.  In section \ref{sec.applications} we discuss two possible applications of the theory. The first one is the construction of quantum tests for various non-standard discrimination problems such as distinguishing between a given state and everything else. The second application is in quantum metrology where symmetric r-body interaction hamiltonians can be used to beat 
the $n^{-1}$ Heisenberg limit in estimation precision. These applications will be investigated in more detail elsewhere.


\section{Classical U-statistics}
\label{sec.classical}

In this section we give a short introduction to the theory of $U$-statistics 
\cite{vanderVaart}. Throughout, $X_{1},\dots ,X_{n}$ will be a random sample from an unknown distribution $\mathbb{P}$ over the measure space $(\Omega, \Sigma)$. 
We consider the problem of estimating the parameter 
$$
\theta= \mathbb{E}(h(X_{1}, \dots, X_{r}))
$$ 
where $h$ is a known function which can be taken to be symmetric with respect to permutations of the arguments. By definition $h(X_{1}, \dots ,X_{r})$ is an unbiased estimator of 
$\theta$ but uses only the first $r$ variables, so we would like to replace it by a better estimator based on the whole sample, where $n>r$. 

\begin{definition}
Let $X_{1},\dots, X_{n}$ be a random sample from an unknown distribution. 
Let 
$
h(x_{1},\dots, x_{r})
$
be a kernel which is invariant under permutations of the arguments. The U-statistic with kernel $h$ is defined as 
$$
U_{n}= {n \choose r}^{-1} \sum_{\beta} h(X_{\beta_{1}} , \dots ,X_{\beta_{r}} )
$$   
where the sum is taken over all unordered subsets $\beta$ of integers from 
$\{1,\dots,n\}$.
\end{definition}
Since the observations are i.i.d. $U_{n}$ is an unbiased estimator of $\theta$ and has smaller variance than $h(X_{1},\dots, X_{r})$. In fact for non-degerate kernels 
$\sqrt{n}(U_{n}-\theta)$ is asymptotically normal.

\begin{theorem}
Let $h(X_{1},\dots , X_{r})$ be a symmetric kernel with 
$\mathbb{E} h^{2}(X_{1} ,\dots X_{r})<\infty$. Let 
$$
h_{1}(x):= \mathbb{E}(h(x, X_{2},\dots , X_{r}))- \theta
$$
and suppose that $\xi_{1}= \mathbb{E} h_{1}^{2}(X_{1})\neq 0$. Then 
$$
\sqrt{n}(U_{n}-\theta)\overset{\mathcal{L}}{\longrightarrow} N(0,r^{2}\xi_{1}).
$$ 
\end{theorem}

\begin{ex}
A $U$-statistic of degree $r=1$ is a mean $\frac{1}{n}\sum_{i=1}^{n} h(X_{i})$ and the above statement is the Central Limit Theorem.
\end{ex}

\begin{ex}
Let $h(x_{1},x_{2})= (x_{1}-x_{2})^{2}/2$ be a kernel of degree $r=2$. In this case
$$
\theta= \frac{1}{2}\mathbb{E}(X_{1}-X_{2})^{2} =  \mathbb{E}(X^{2}) -
(\mathbb{E}(X))^{2}= {\rm Var}(X).
$$ 
The corresponding $U$-statistic is the sample variance
$$
U_{n}= \binom{n}{2}^{-1} \sum_{i<j} \frac{1}{2} (X_{i}- X_{j})^{2}= 
\frac{1}{n-1} \sum_{i=1}^{n} (X_{i}^{2}- \bar{X}^{2}),
$$
where $\bar{X}:= \frac{1}{n}\sum_{i=1}^{n}X_{i}$. 
We have 
$
h_{1}(x)= \frac{1}{2}\left[  (x- \mathbb{E}(X))^{2} - {\rm Var}(X)\right]
$ 
and $\xi_{1}= {\rm Var}((X-\mathbb{E}(X))^{2})/4\neq 0$ unless $X-\mathbb{E}(X)$ has a distribution which is concentrated on two points $\{a, -a\}$.
\end{ex}

If $\xi_{1}=0$ then the convergence to normal distribution does not hold, and one calls the kernel $h$ {\it degenerate}. In such cases the limit distribution is that of a polynomial in Gaussian variables. The proof uses an elegant tool called the Hoeffding decomposition which we describe below.

Let $L^{2}(\Omega^{n}, \Sigma^{n}, \mathbb{P}^{n})$ be the space of square 
integrable functions of $(X_{1},\dots, X_{n})$. For every subset $B\subset \{1,\dots ,r\}$ we identify the subspace $\mathcal{L}_{B}:=L^{2}(\Omega^{B}, \Sigma^{B}, \mathbb{P}^{B})$ consisting of functions depending only on the variables $(X_{i}:i\in B)$. The projection $Q_{B}$ onto this subspace is simply the conditional expectation
$$
Q_{B} f = \mathbb{E}(f | X_{i}: i\in B ).
$$
By construction, these subspaces are not mutually orthogonal but respect the partial order defined by set inclusion
$$
B^{\prime}\subset B \Longrightarrow \mathcal{L}_{B^{\prime}} \subset \mathcal{L}_{B}.
$$
In order to compute covariances, it would be more convenient to associate to each subset $A$ a subspace $\mathcal{H}_{A}$ which contains only those vectors in   
$\mathcal{L}_{A}$ which are orthogonal onto all smaller subspaces 
$\mathcal{L}_{B^{\prime}}$.  Thus $\mathcal{H}_{A}$ consists of functions 
$g(X_{i}: i\in A)$ such that
$$
\mathbb{E}(g(X_{i}:i\in A ) | X_{j}: j\in B)=0, {\rm for~all~} B: |B|<|A| .
$$
Here are the simplest examples of how the projection $P_{A}$ onto $\mathcal{H}_{A}$ acts:
\begin{eqnarray*}
P_{\emptyset}f &=& \mathbb{E}(f),  \\
P_{\{i\}} f &=& \mathbb{E}(f|X_{i})- \mathbb{E}(f)\\
P_{\{i,j\}} f &=& \mathbb{E}(f|X_{i}, X_{j})- \mathbb{E}(f|X_{i})- \mathbb{E}(f|X_{j})+ \mathbb{E}(f)
\end{eqnarray*}

\begin{theorem}[Hoeffding decomposition]
Let $f\in L^{2}(\Omega^{n}, \Sigma^{n}, \mathbb{P}^{n})$. Then the projection of $f$ onto $\mathcal{H}_{A}$ is given by
$$
P_{A}f = \sum_{B\subset A} (-1)^{|A|-|B|} \mathbb{E}(f| X_{i}: i\in B).
$$ 
The spaces $\mathcal{H}_{A}$ are mutually orthogonal and span 
$L^{2}(\Omega^{n}, \Sigma^{n}, \mathbb{P}^{n})$.
\end{theorem}
A simple way to visualise the subspaces $\mathcal{H}_{A}$ is to construct an orthonormal basis. Let  $ 1=f_{0}, f_{1},\dots$ be an orthonormal basis of $L^{2}(\Omega, \Sigma, \mathbb{P})$. Then $\mathcal{H}_{A}$ is the span of basis vectors of the form
$$
f_{i_{1}}\otimes \dots\otimes  f_{i_{n}}
$$
for which $i_{k}\geq 1$ if and only if $k\in A$.

We now apply the Hoeffding decomposition to $U$-statistics. A combinatorial argument shows that
$$
U_{n}= \binom{n}{r}^{-1}\sum_{l=0}^{r} \sum_{|A|=l}\sum_{\beta} 
P_{A} h(X_{\beta_{1}}, \dots , X_{\beta_{r}}) = \sum_{l=0}^{r} \binom{r}{l}U_{n,l} 
$$
where $U_{n,l}$ is the $U$-statistic corresponding to the symmetric kernel $h_{l}$ of order $l$, equal  to the projection
\begin{equation}\label{eq.h_l}
h_{l}:= P_{\{1,\dots, l\}}h.
\end{equation}
\begin{lemma}\label{lemma.variance.un}
Let $h(x_{1},\dots, x_{r})$ be a kernel of of order $r$ and let $U_{n}$ be the associated $U$-statistic. Then
$$
{\rm Var}(U_{n}) =\sum_{l=1}^{r} \binom{r}{l}^{2} \binom{n}{l}^{-1} \mathbb{E}(h_{l}^{2}).
$$
If $h_{1}=\dots= h_{c-1}=0$ and $h_{c}\neq 0$ then ${\rm Var}(U_{n})$ is of order 
$O(n^{-c})$.
\end{lemma}


\begin{proof}
From the definition of $h_{l}$ it follows that $U_{n,i}$ and $U_{n,j}$ are not correlated for $i\neq j$. On the other hand the variance of $U_{n,l}$ can be easily computed as
$$
{\rm Var}(U_{n,l}) = \binom{n}{l}^{-2} \sum_{\beta,\beta^{\prime}} 
\mathbb{E}(h_{l}(X_{\beta_{1}},\dots ,X_{\beta_{l}} )h_{l}(X_{\beta^{\prime}_{l}},\dots ,X_{\beta^{\prime}_{l}} ))= \binom{n}{l}^{-1}\mathbb{E}(h_{l}^{2}(X_{1},\dots , X_{l})) 
$$
which is of order $n^{-l}$. 
Then
$$
{\rm Var}(U_{n})= \sum_{l=1}^{r} \binom{r}{l}^{2}\binom{n}{l}^{-1}
\mathbb{E}(h_{l}^{2}(X_{1},\dots , X_{l})). 
$$
\end{proof}
This lemma shows that the dominant contribution to $U_{n}$ comes from $U_{n,c}$ and suggests that $n^{c/2} U_{n}$ may have a limit distribution as in the non-degenerate case. We will then formulate the limit theorem for the $U$-statistic associated to the kernel $h_{c}$ defined as above.

\begin{theorem}\label{th.convergence.classic}
Let $h_{c}(X_{1},\dots , X_{c})$ be a symmetric kernel such that 
$\mathbb{E}h_{c}^{2}(X_{1},\dots , X_{c})<\infty$ and $\mathbb{E}(h_{c}(x_{1},\dots , x_{c-1}, X_{c}))\equiv 0$. Let $1=f_{0}, f_{1},\dots $ be an orthonormal basis of 
$L^{2}(\Omega, \Sigma, \mathbb{P})$. Then the sequence of $U$-statistics $U_{n,c}$ with kernel $h_{c}$ satisfies
$$
n^{c/2}U_{n,c}\overset{\mathcal{L}}{\longrightarrow} 
\sum_{{\bf k}=(k_{1},\dots,k_{c})\in \mathbb{N}^{c}}
\langle h_{c}, f_{k_{1}}\otimes\dots \otimes f_{k_{c}} \rangle
\prod_{j} H_{a_{j}({\bf k})}(\mathbb{B} (f_{j}) ).
$$ 
Here $H_{a}$ are the Hermite polynomials, $\mathbb{B} (f_{j})$ are independent standard normal variables and $a_{j}({\bf k})$ is the number of times $f_{j}$ occurs among $f_{k_{1}},\dots, f_{k_{c}}$. The variance of the limit variable is equal to 
$c! \mathbb{E}h_{c}^{2}(X_{1},\dots , X_{c})$.
\end{theorem}

\section{Thermal equilibrium states of a harmonic oscillator}
\label{sec.oscillator}

In the classical case, the limit of $U$-statistics is described in terms of polynomials in normally distributed variables. The simplest quantum analogue is a quasifree (Gaussian) state a quantum harmonic oscillator. The first part of this section gives a short overview of the basic notions needed in the main result and for the formulation of the Quantum Central Limit Theorem.

In the second part we prove a technical result which will be employed in the convergence theorem of $U$-statistics: the symmetric product of Hermite polynomials in the canonical variables is orthogonal to all lower order polynomials, where the inner product is given by the covariance with respect to a 
thermal equilibrium state.

A variety of quantum systems such as the free quantum particle, the quantum harmonic oscillator, the monochromatic light beam are described mathematically by the same algebra of observables generated by the `canonical coordinates' $Q$ and $P$ satisfying the Heisenberg's commutation relations
\begin{equation}\label{eq.Heisenberg}
[Q,P]:= QP-PQ = i\mathbf{1}.
\end{equation}
These observables can be represented as (unbounded) selfadjoint operators on $L^{2}(\mathbb{R})$
$$
Q\psi(q) = q\psi(q),\qquad P\psi(q) = \frac{1}{i} \frac{d\psi}{dq} ,
$$
where $\psi$ is a vector in the domain of the respective observable. The space $L^{2}(\mathbb{R})$ has a special orthonormal basis 
$\{  \left\lvert  0 \right\rangle,  \left\lvert  1 \right\rangle, \dots\}$ with the vector 
$ \left\lvert  m \right\rangle$ given by 
$$
H_{m}(q) e^{-q^{2}/2} / (\sqrt{\pi}2^{m}m!)^{1/2},
$$ 
where $H_{m}$ are the Hermite polynomials \cite{Erdelyi}. These are the eigenvectors of the number operator ${\bf N}:= \frac{1}{2}(Q^{2}+ P^{2}-\mathbf{1})$ counting the number of `excitations' of the oscillator or the number of photons in the case of the light beam, such that ${\bf N}  \left\lvert  m \right\rangle=m  \left\lvert  m \right\rangle$.

In order to avoid technicalities related to the fact that $Q$ and $P$ are unbounded operators, one can work instead with the unitary operators $U(a):= \exp(iaQ) $ and $V(b):= \exp(ibP)$ and encode \eqref{eq.Heisenberg} into the following relation
$$
U(a)V(b)= \exp (iab)\,V(b)U(a).
$$
By combining the two families of unitaries into a single {\it Weyl operator}
$$
S(a,b):= \exp(-\frac{1}{2}iab)\, U(a)V(b), \qquad (a,b)\in \mathbb{R}^{2},
$$
we obtain a projective unitary representation of the abelian group 
$\mathbb{R}^{2}$ 
\begin{equation}\label{eq.wey.relation}
S(a,b)S(c,d) = \exp\left(\frac{i}{2}(ad-bc)\right) S(a+c,b+d).
\end{equation}
Since the algebra generated by the Weyl operators is dense in 
$\mathcal{B}(L^{2}(\mathbb{R}))$ with respect to the weak topology, we will refer to the 
latter as the algebra of observables of the harmonic oscillator.

The pure state $|0\rangle\langle 0|$ is called the {\it vacuum} or {\it zero temperature} state and satisfies
$$
\langle 0| S(a,b)| 0\rangle = \exp\left(- \frac{1}{4}(a^{2}+ b^{2})\right). 
$$ 
From this it follows that in the vacuum both $Q$ and $P$ have centred normal distributions known as `vacuum fluctuations', with variances $\langle 0| Q^{2}| 0\rangle=\langle 0| P^{2} | 0\rangle=1/2$. Their product is equal to the minimum of $1/4$ allowed by the Heisenberg's uncertainty principle.

Besides vacuum, we will be interested in {\it thermal equilibrium} states 
\begin{equation}\label{eq.thermal.state.1}
\varphi_{\sigma}(S(a,b))= \exp\left(- \frac{\sigma^{2}}{2}(a^{2}+ b^{2}) \right), \qquad \sigma^{2}> \frac{1}{2},
\end{equation}
where the variances of $Q,P$ depend on the inverse temperature $\beta$ as $\sigma^{2}= (2\tanh(\beta/2))^{-1}$. Thermal equilibrium states are mixed and their 
density matrix  $\phi_{\sigma}$ is a diagonal in the number operator basis
\begin{equation}\label{eq.thermal.state.2}
\phi_{\sigma}= (1-e^{-\beta}) \sum_{k=0}^{\infty}  |k\rangle \langle k| e^{-k\beta} .
\end{equation}

Both vacuum and thermal states are particular examples of quasifree states, the latter being characterised by the fact that all linear combinations of canonical variables have Gaussian distributions. Any quasifree state can be obtained by 
applying a sequence of unitary operations called displacement, phase transformation and squeezing to a thermal or vacuum state \cite{Leonhardt}.

Since the Hermite polynomials are orthogonal with respect to the normal distribution 
$N(0,1/2)$ we have 
\begin{equation}\label{eq.ortho.vacuum}
\langle 0|  H_{n}(Q) H_{m}(Q) |0 \rangle =
\frac{1}{\sqrt{\pi}}\int H_{n}(x)H_{m}(x)e^{-x^{2}}  =  \delta_{n,m} 2^{m}m!
\end{equation}

Let $\phi_{\sigma}$ be the thermal equilibrium state defined in 
\eqref{eq.thermal.state.1} and  \eqref{eq.thermal.state.2}, and let $L^{2}(\phi_{\sigma})$ be the completion of $\mathcal{B}(L^{2}(\mathbb{R}))$ with respect to the complex inner product
\begin{equation}\label{eq.orthogonality.hermite}
\langle A, B \rangle_{\phi_{\sigma}} = {\rm Tr}(\phi_{\sigma} A^{*}B).
\end{equation}
By a change of variable \eqref{eq.ortho.vacuum} can be rewritten as  
$$
\langle H_{n}(Q/\sqrt{2}\sigma)), \,H_{m}(Q/\sqrt{2}\sigma))\rangle_{\phi_{\sigma}}= 
\delta_{n,m} 2^{m}m!, 
$$
and a similar orthogonality relation can be derived for $P$. 
Since in the commutative case $H_{n}(x)H_{m}(y)$ is a basis in 
$L^{2}(\mathbb{R}^{2} , \exp(-(x^{2}+y^{2}))$, this suggests that 
$L^{2}(\phi_{\sigma})$ may have an orthogonal basis consisting of some type of products of Hermite polynomials in $Q$ and $P$.


%

\begin{definition}
Let $Q,P$ be the canonical observables of a quantum harmonic oscillator. The symmetric ordering on polynomials in  $Q,P$ is the operation defined by replacing each monomial 
$
X_{1}\dots X_{n}
$  
by 
$$
\mathcal{S}[X_{1}\dots X_{n}]:= \frac{1}{n!}\sum_{\tau\in S(n)} X_{\tau(1)}\dots X_{\tau(n)},
$$
where $X_{i}$ is either $Q$ or $P$.
\end{definition}

The following theorem shows that the symmetric products of Hermite polynomials behave similarly to their classical counterparts in what concerns orthogonality between polynomials of {\it different} orders. However for same order polynomials it can be directly checked that the orthogonality relation does not hold.

\begin{theorem}\label{th.orthog.hermite.noncomm}
Let $\langle\cdot, \cdot\rangle_{\phi_{\sigma}} $  be the inner product 
\eqref{eq.orthogonality.hermite} on $\mathcal{B}(L^{2}(\mathbb{R}))$. Then
$$
\langle \mathcal{S}[H_{n}(Q/ \sqrt{2}\sigma)H_{m}(P/ \sqrt{2}\sigma)] , \,Y\rangle_{\phi_{\sigma}} = 0, 
$$
for any polynomial $Y$ in $Q, P$ of order strictly smaller than $n+m$. 

In particular, 
$\mathcal{S}[H_{n}(Q/ \sqrt{2}\sigma)H_{m}(P/ \sqrt{2}\sigma)]$ is the unique projection of $(\sqrt{2}\sigma)^{-n-m}\mathcal{S}[Q^{n}P^{m}]$ onto the orthogonal complement of polynomials of order strictly smaller than $n+m$.
\end{theorem}

\begin{proof}
Considering the analogy with the commutative case, it is natural to use the theory of Wigner functions for computing the inner product. We briefly review here the necessary background and refer to \cite{Leonhardt} for more details.

To any trace-class operator $A\in \mathcal{T}_{2}(L^{2}(\mathbb{R}))$ we associate its Wigner function 
$W_{A}\in L^{2}(\mathbb{R}^{2})$ defined as the inverse Fourier transform with respect to both variables of 
$$
\widetilde{W}_{A}(u,v) = {\rm Tr}(A \exp(-i (uQ +vP)). 
$$

1. The map $W:A\to W_{A}$ is an isomorphism between $\mathcal{T}_{2}(L^{2}(\mathbb{R}))$ and $L^{2}(\mathbb{R}^{2}) $ and
$$
{\rm Tr}(A\, B) = 2\pi \int\!\!\!\!\int W_{A}(q,p)W_{B}(q,p) \,dq dp.
$$ 

2. The expectations of symmetric polynomials in $Q,P$ can be computed as
\begin{equation}\label{eq.isometry.wigner}
{\rm Tr}(A \mathcal{S}[Q^{n}P^{m}])= 2\pi\int\!\!\!\!\int W_{A} (q,p) q^{n}p^{m} \, dq dp
\end{equation}

3. The Wigner function of a product $AB$ can be computed as \cite{Pool}
\begin{equation}\label{eq.product.weyl}
W_{AB}(q,p) = \frac{1}{4\sqrt{2\pi}} \int\!\!\!\!\int\!\!\!\!\int\!\!\!\!\int
W_{A}(q^{\prime}, p^{\prime}) W_{B} (q^{\prime\prime},p^{\prime\prime})
\exp\left[ -2i \,{\rm Det} 
\left| 
\begin{array}{ccc}
1 & q & p\\
1 & q^{\prime} & p^{\prime} \\
1 & q^{\prime\prime} & p^{\prime\prime}
\end{array}
\right|
 \right] 
dq^{\prime} dp^{\prime} dq^{\prime\prime} dp^{\prime\prime}.
\end{equation}

Inserting $W_{\phi_{\sigma}} (q,p)= \exp(-(q^{2}+p^{2})/2\sigma^{2})/2\pi\sigma^{2}$ and $W_{Y}(q,p)= Y(q,p)$ into \eqref{eq.product.weyl} we get
\begin{eqnarray*}
W_{Y\phi_{\sigma}} (q,p) &=&
\frac{1}{4\sigma^{2}(2\pi)^{3/2}} \int\!\!\!\!\int\!\!\!\!\int\!\!\!\!\int 
Y(q^{\prime},p^{\prime})  \exp(-(q^{\prime\prime 2}+p^{\prime\prime 2})/2\sigma^{2})
\exp\left[ -2i \,{\rm Det} 
\left| 
\begin{array}{ccc}
1 & q & p\\
1 & q^{\prime} & p^{\prime} \\
1 & q^{\prime\prime} & p^{\prime\prime}
\end{array}
\right|
 \right] 
dq^{\prime} dp^{\prime} dq^{\prime\prime} dp^{\prime\prime}\\
&=&
\frac{1}{4\sqrt{2\pi}} \int\!\!\!\!\int
Y(q^{\prime},p^{\prime}) \exp(-2i (qp^{\prime} -q^{\prime}p) ) 
\exp(-2\sigma^{2}( ( p-p^{\prime})^{2}+ (q-q^{\prime})^{2} )dq^{\prime} dp^{\prime}.
\end{eqnarray*}

Any polynomial $Y(q^{\prime},p^{\prime})$ 
of order strictly smaller than $m+n$ can be expresses as a linear combination with coefficients $\alpha(a,b,c,d)$ independent of $(q,p,q^{\prime},p^{\prime})$ of terms
$$
T (q^{\prime},p^{\prime}):= q^{c}p^{d}
H_{a}\left(\sqrt{2}\sigma(q^{\prime}-q) \right) 
H_{b}\left(\sqrt{2}\sigma(p^{\prime}-p)\right), \qquad  a+b+c+d<m+n.
$$ 
In this case the above integral can be computed explicitly by using the following property of Hermite polynomials \cite{Erdelyi} 
$$
H_{k}(q) e^{-q^{2}}=(-1)^{n} \frac{d^{k}}{dq^{k}} \exp(-q^{2}).  
$$
By a change of variables we have
\begin{eqnarray*}
W_{T\phi}(q,p)&=&
\frac{1}{8\sigma^{2}\sqrt{2\pi}} 
\int\!\!\!\int
H_{a}(q^{\prime}) H_{b}(p^{\prime})\exp(-q^{\prime 2}-p^{\prime 2}) 
\exp(-2i (qp^{\prime}- pq^{\prime} )/\sqrt{2}\sigma) dq^{\prime} dp^{\prime}\\
&=&
C q^{a+c}p^{b+d} \exp(-(q^{2}+ p^{2})/2\sigma^{2}),
\end{eqnarray*}
for some constant $C= C(a,b,c,d,\sigma)$. Now using \eqref{eq.isometry.wigner} we get
\begin{eqnarray*}
&&
{\rm Tr}(\phi \mathcal{S}[H_{n}( Q/\sqrt{2}\sigma) H_{m}(P/\sqrt{2}\sigma)] Y) \\
&&=
2\pi \int\!\!\!\!\int 
H_{n}(q/\sqrt{2}\sigma) H_{m}(p/\sqrt{2}\sigma) W_{Y\phi} (q,p)\, dqdp\\
&&=
2\pi \sum \alpha(a,b,c,d) C^{\prime}(a,b,c,d,\sigma)
\int\!\!\!\!\int
H_{n}(q ) H_{m}(p) p^{b+d}q^{a+c} \exp(-( q^{2}+ p^{2})) \, 
dqdp=0,
\end{eqnarray*}
where in the last line we used the following property of Hermite polynomials 
$$
\int H_{k}(q)  q^{c}  \exp(- q^{2})dq= 0, \qquad {\rm for }~ c<k. 
$$

The second statement of the theorem follows from the fact that the leading term of the polynomial  $\mathcal{S}[H_{n}(Q/\sqrt{2}\sigma) H_{m}(P/\sqrt{2}\sigma)]$ is 
$(\sqrt{2}\sigma)^{-n-m}\mathcal{S}[Q^{n}P^{m}]$.

\end{proof}

\section{Quantum Central Limit Theorem}\label{sec.clt}

In this section we give a brief review the quantum Central Limit Theorem (CLT) \cite{Petz}. Similarly to its classical counterpart, the CLT shows how certain Gaussian states of quantum harmonic oscillators emerge as the limit of `fluctuations' of collective observables of a large number of identically prepared systems. Let $\Psi$ be a faithful state on $M(\mathbb{C}^{d})$, i.e. 
$$
\Psi(A) = {\rm Tr}(\rho A), \qquad A\in M(\mathbb{C}^{d}),
$$
with $\rho$ a strictly positive density matrix. To $(M(\mathbb{C}^{d}), \rho)$ 
we will associate an algebra of canonical commutation relations carrying a Gaussian state $\varphi$. By `diagonalisation', the algebra can be easily identified with the tensor product of $d(d-1)/2$ quantum harmonic oscillators  and the commutative algebra 
$C_{b}(\mathbb{R}^{d-1})$. The state $\varphi$ factorises as well into a product of thermal equilibrium states and a $(d-1)$-dimensional Gaussian distribution.

\begin{definition}
Let $M(\mathbb{C}^{d})_{sa}$ be the space of selfadjoint matrices. On 
$M(\mathbb{C}^{d})_{sa}$ we defined the (positive) inner product 
$$
(A, B)_{\rho} := \mathrm{Tr}(\rho\, A\circ B), \qquad {\rm where~}\qquad  A\circ B 
:= \frac{AB + BA}{2},
$$
and we denote by $L^{2}(\rho) $ the Hilbert space 
$\left( M(\mathbb{C}^{d})_{sa},(\cdot, \cdot)_{\rho}\right)$. By 
$L^{2}_{0}(\rho)$ we denote the orthogonal complement of the identity, i.e.
$
L^{2}_{0}(\rho)= \{ A \in L^{2}(\rho) : (A, \mathbf{1})_{\rho}= {\rm Tr}(\rho A)=0\}.
$   

We further define the {\it symplectic form} $\sigma$ on $L^{2}(\rho)$
$$
\sigma(A,B): = \frac{i}{2}\mathrm{Tr}(\rho\, [A, B]).
$$

\end{definition}

Roughly speaking, the algebra of canonical commutation relations 
$CCR(L^{2}_{0}(\rho), \sigma)$ is generated by the canonical variables 
$\mathbb{B}(A)$ satisfying the conditions
\begin{eqnarray*}
\mathbb{B}(A+B) &=& \mathbb{B}(A)+ \mathbb{B}(B),  \\
\mathbb{B}(A) &=& \mathbb{B}(A)^{*},\\
\mathbb{B}(\lambda A) &=& \lambda \mathbb{B}(A), \qquad \lambda\in \mathbb{R},\\
  \,[ \mathbb{B}(A),\mathbb{B}(B)] 
&=&  i2\sigma(A,B)\mathbf{1},
\end{eqnarray*}
for all $A,B \in L^{2}_{0}(\rho)$. The harmonic oscillator is a particular case where the space of canonical variables is spanned by $Q$ and $P$ satisfying the commutation relation 
$[Q,P]= i\mathbf{1}$. 

Alternatively, one can define the unitary Weyl operators 
$S(A):= \exp(i\mathbb{B}(A))$ which satisfy 
$$
S(A)^{*} = S(-A), \qquad S(A)S(B) = S(A+B)\exp(-i\sigma(A,B)), \quad A,B \in 
L^{2}_{0}(\rho).
$$
%
On $CCR(L^{2}_{0}(\rho), \sigma)$ we define the quasifree state 
\begin{equation}\label{eq.quasifree}
\varphi (S(A)) := \exp\left(-\frac{1}{2} \| A\|_{\rho}^{2}\right), \qquad \| A\|_{\rho}^{2}= (A,A)_{\rho},
\end{equation}
with respect to which $\mathbb{B}(A)$ has normal distribution $N(0,\|A\|_{\rho}^{2})$.

We will now `diagonalise' $CCR(L^{2}_{0}(\rho), \sigma)$ and construct an explicit Hilbert space representation of the algebra and the state  $\varphi$. Let 
\begin{equation}\label{eq.rho.diagonal}
\rho= \sum_{i=1}^{d} \lambda_{i} |e_{i}\rangle\langle e_{i}|, 
\end{equation}
be the spectral decomposition of the density matrix $\rho$ and assume that 
$\lambda_{1}> \dots >\lambda_{d}$. The Hilbert space $L^{2}_{0}(\rho)$ decomposes as direct sum of orthogonal subspaces
$\mathcal{H}_{\rho} \oplus \mathcal{H}_{\rho}^{\perp} $ where
\begin{equation}\label{eq.orthog.decomp}
\mathcal{H}_{\rho} := 
{\rm Lin} \{ A\in L^{2}_{0}(\rho) : [A,\rho] =0 \},
\quad
{\rm and} 
\quad
\mathcal{H}_{\rho}^{\perp}=  {\rm Lin} \{ T_{j,k} , 1\leq j\neq k\leq d\},
\end{equation}
where $T_{j,k}$ are the matrices
\begin{eqnarray}\label{generators_algebra}
T_{j,k} &= & i E_{j,k} - i E_{k,j}  \qquad \text{for} ~ 1\leq j<k \leq d; \nonumber \\
T_{k,j} &=  & E_{j,k} + E_{k,j}   \qquad \text{for}     1\leq j<k \leq d.
\end{eqnarray}  
with $E_{i,j}$ the matrix with entry $(i,j)$ equal to $1$, and all others equal to $0$.

The elements $S(A)$ with $A\in \mathcal{H}_{\rho}$ generate the center of the CCR
algebra, which is isomorphic to the algebra  of bounded continuous functions $C_{b}(\mathbb{R}^{d-1})$. Explicitly, we identify the coordinates in 
$\mathbb{R}^{d-1}$ with the basis 
$\{ d_{i}:= -\mu_{i}\mathbf{1} + E_{i,i}: i=1,\dots, d-1\} $ of $\mathcal{H}_{\rho}$. 
Then the covariance matrix of the `classical' Gaussian variables 
$G_{i}:=\mathbb{B}(d_{i})$ is 
$$
( d_{i}, d_{j})_{\rho} = 
{\rm Tr}(\rho d_{i} d_{j} ) = \delta_{i,j}\mu_{i} - \mu_{i}\mu_{j} := 
[V(\lambda)]_{i,j}.
$$ 
As for the space $\mathcal{H}_{\rho}^{\perp}$, we identify basis
\begin{equation}\label{eq.symplectic.basis}
t_{j,k}:= T_{j,k}/\sqrt{2|\mu_{j}-\mu_{k} |}, \qquad j\neq k,
\end{equation}
which is both {\it orthogonal and symplectic} i.e. 
$$
\sigma(t_{j,k} , t_{k,j}) = -1/2  ,\quad  j<k ,\quad {\rm and}~ ~ \sigma(t_{j,k} , t_{l,m})= 0 \quad {\rm for ~} \{ j,k\} \neq \{l,m\}.
$$
This means that for each $(j\neq k)$ the operators
$$
Q_{j,k}:=\mathbb{B}(t_{j,k}),\qquad P_{j,k}:= \mathbb{B}(t_{k,j}) 
$$ 
form a pair of canonical coordinates of a harmonic oscillator, and all such 
$d(d-1)/2$ oscillators commute with each other. 
Moreover, since the basis is orthogonal, the state $\varphi$ factorises into a product of thermal equilibrium states $\varphi_{\sigma_{j,k}}$ with variances 
\begin{equation}\label{eq.variance.thermal}
\sigma_{j,k}^{2}= \varphi(Q_{j,k}^{2})= \varphi(P_{j,k}^{2}) =
\| t_{j,k}\|_{\rho}^{2} = {\rm Tr}(\rho t_{j,k}^{2}) = 
\frac{\mu_{j} + \mu_{k}}{2|\mu_{j}-\mu_{k}|} >\frac{1}{2}.
\end{equation}

\begin{remark}
The above decomposition assumes that $\rho$ has a non-degenerate spectrum. It is easy to verify that a similar result holds in the general case, with some of quantum oscillators being replaced by bivariate classical normal variables.  
\end{remark}

We pass now to the formulation of the Quantum CLT. 
Consider the tensor product  $\bigotimes_{k=1}^{n}M(\mathbb{C}^{d})$ which is generated by elements of the form
\begin{equation}\label{eq.xk}
A^{(k)} = \mathbf{1}\otimes \dots \otimes A \otimes \dots \otimes \mathbf{1},
\end{equation}
with $A$ acting on the $k$-th position of the Hilbert space tensor product 
$\left(\mathbb{C}^{d}\right)^{\otimes n}$. We are interested in the asymptotics as $n\to\infty$ of the joint distribution under the state $\rho^{\otimes n}$, of `fluctuation' elements of the form
\begin{equation}\label{eq.fluctuation.obs}
\mathbb{F}_{n}(A) :=\frac{1}{\sqrt{n}} \sum_{k=1}^{n} A^{(k)}.
\end{equation}
\begin{theorem}{\bf [Quantum CLT]}\label{th.clt}
Let $A_{1}, \dots , A_{s}\in M(\mathbb{C}^{d} )_{sa}$ such that 
${\rm Tr}(\rho A_{i})=0$. Then 
\begin{eqnarray*}
&&
\lim_{n\to\infty} {\rm Tr}
\left(\rho^{\otimes n} \left(\prod_{l=1}^{s} \mathbb{F}_{n}(A_{l}) \right)\right) =
\varphi \left( \prod_{l=1}^{s}\left( \mathbb{B}(A_{l}) \right)\right),\\
&&
\lim_{n\to\infty} {\rm Tr} 
\left( \rho^{\otimes n} \left( \prod_{l=1}^{s} \exp( i \mathbb{F}_{n}(A_{l}) ) \right)\right) =
\varphi\left( \prod_{l=1}^{s} S( A_{l} )  \right).
\end{eqnarray*}
\end{theorem}

For later purposes we need the following corollary which combines the CLT with the Law of Large Numbers.
Let $C\in M(\mathbb{C}^{d})_{sa}$ with expectation $c= {\rm Tr}(\rho C)$ 
and define the `empirical average'
\begin{equation}\label{eq.empirical.average}
\mathbb{P}_{n}(C):= \frac{1}{n} \sum_{k=1}^{n}C^{(k)} = 
c\mathbf{1} +  \frac{1}{\sqrt{n}} \mathbb{F}_{n}(C-c\mathbf{1})
\end{equation}

\begin{corollary}\label{cor.1}
Let $A_{1},\dots , A_{s}\in M(\mathbb{C}^{d} )_{sa}$ such that 
${\rm Tr}(\rho A_{i})=0$ and 
let $C_{1},\dots C_{l} \in M(\mathbb{C}^{d} )_{sa}$ with expectations 
$c_{i}:={\rm Tr}(\rho C_{i})$. For any non-commutative polynomial 
$P[X_{1}, \dots X_{s}, Y_{1}, \dots, Y_{l} ]$ in the free variables $X_{i}, Y_{j}$ 
we have
\begin{eqnarray*}
\lim_{n\to\infty} {\rm Tr}
\left(\rho^{\otimes n} P[\mathbb{F}_{n}(A_{1}) , \dots , \mathbb{F}_{n}(A_{s}) , \mathbb{P}_{n}(C_{1}), \dots, \mathbb{P}_{n}(C_{l}) ] \right)= 
\varphi \left( P[ \mathbb{B}(A_{1}), \dots, \mathbb{B}(A_{k}), c_{1}\mathbf{1}, \dots, c_{l}\mathbf{1}]
\right)
\end{eqnarray*}

\end{corollary}

\begin{proof}
Write $ \mathbb{P}_{n}(C)$ as in \eqref{eq.empirical.average} and apply CLT while noting that any monomial containing at least one term $\mathbb{F}_{n}(C_{i}-c_{i}\mathbf{1})$ converges to zero. 
\end{proof}

\begin{remark}
The quantum CLT provides two convergence results, for moments and for exponentials. In the commutative case, the two are equivalent and are also equivalent to convergence in distribution \cite{Billingsley}. This follows from L\'{e}vy's continuity Theorem (convergence in distribution is equivalent to the convergence of the characteristic functions) and the fact that convergence in distribution is equivalent to convergence in moments, provided that the limit distribution is completely characterised by its moments.
%
In the quantum case we cannot talk about a joint distribution of non-commuting variables and the two results seem to be unrelated. 
\end{remark}

For later purposes we formulate now the extension of Theorem 
\ref{th.orthog.hermite.noncomm} to the case of the algebra 
$CCR(L^{2}_{0}(\rho))$. We denote by $\{F_{1},\dots F_{d^{2}-1}\}$ the orthonormal basis of $L^{2}_{0}(\rho)$ defined earlier 
$$
\{d_{i} : i=1,\dots d-1\} \cup \{t_{j,k}: 1\leq j\neq k\leq d\}.
$$
arranged in a particular order.
\begin{corollary}\label{cor.hermite.orthogonal}
Let $\mathcal{P}_{c}\subset L^{2}(\varphi)$ be the linear span of polynomials of order up to $c$ in the canonical variables $\mathbb{G}(F_{i})$, and denote by 
$Z_{c}$ its orthogonal projection. Then the symmetric product of (non-commuting) Hermite polynomials
$$
H_{\bf m}:=\mathcal{S}\left[ \prod_{i} H_{m_{i}}( \mathbb{G}(F_{i})) \right],
$$
is orthogonal on $\mathcal{P}_{c-1}$, for any set ${\bf m}:=\{ m_{i}\}$ of multiplicites with $|{\bf m}|=c$. In particular, 
$$
Z_{c-1}^{\perp} \mathcal{S}\left [\prod_{i} \mathbb{G}(F_{i})^{m_{i}} \right] =
H_{\bf m}.
$$

\begin{proof}

The result follows form Theorem \ref{th.orthog.hermite.noncomm} applied 
to each of the quantum oscillators $(Q_{j,k}, P_{j,k})$ in $CCR(L^{2}_{0}(\rho))$ together with the classical counterpart for the `classical' canonical variables $G_{i}$.

Indeed by writing the symmetric operators and the projections in terms of tensor products over the $d(d-1)/2$ oscillators and $(d-1)$ independent normal variables, we can derive the corrollary from the following lemma.

\end{proof}
\begin{lemma}
Let 
$$
\mathcal{H}_{a} = \bigoplus_{i\in \mathbb{N}} \mathcal{H}_{a,i} , \qquad a=1\dots k,
$$
be arbitrary decompositions of Hilbert spaces $\mathcal{H}_{a}$ into orthogonal subspaces with projections $P_{a,i}$. Let 
$$
Q_{a,i}= \bigoplus_{b=1}^{a} P_{b,i}
$$
 be the projections onto vectors in $\mathcal{H}_{a}$ of `order less or equal to $a$'. Similarly let $Z_{c}$ be the projection onto the vectors of the tensor product 
 $\mathcal{H}_{1}\otimes\dots\otimes \mathcal{H}_{k}$ 
 of total order less or equal to $c$ 
$$
Z_{c}:= \bigoplus_{\{ m_{a} \}} P_{1, m_{1} } \otimes \dots \otimes P_{k ,m_{k} },
$$
where the direct sum runs over all multiplicities ${\bf m}={\{ m_{a} \}}$ with 
$|{\bf m}|\leq c$. 

Then for any $k$-tuple $\{\psi_{a}= Q_{m_{a}} \psi_{a} :a=1,\dots k\}$ with $|{\bf m}|=c$ we have
$$
Z_{c-1}^{\perp} \left( \bigotimes_{a} \psi_{a} \right)=  
\bigotimes_{a} P_{a,m_{a}}\psi_{a}= \bigotimes_{a} Q_{a,m_{a}-1}^{\perp}\psi_{a}.
$$

\end{lemma}

\end{corollary}

\section{The Hoeffding decomposition}\label{sec.q.hoeffding}
This section deals with the generalisation of the concept of Hoeffding decomposition to non-commuting variables. When seen as a Hilbert space construction, the difference between classical and quantum cases disappears and one obtains the same result as the classical one. However we 
include the proofs for the sake of clarity.

As before, given $H\in M(\mathbb{C}^{d})$ and any $1\leq i\leq n$ we denote by 
$H^{(i)}$ the operator on $(\mathbb{C}^{d})^{\otimes n}$
\begin{equation}\label{eq.action.position.i}
H^{(i)} := \mathbf{1} \otimes \dots \otimes\mathbf{1} \otimes H\otimes \mathbf{1} \otimes
\dots \otimes \mathbf{1}
\end{equation}
with $H$ acting on the $i$-th term of the tensor product.

\begin{definition}
Let $\rho^{\otimes n}\in M(\mathbb{C}^{d})^{\otimes n}$ be the joint state of $n$ independent identically prepared systems. Let $A$ be a subset of 
$\{1,\dots ,n\}$. The conditional expectation 
$$
\mathbb{E}(\cdot | A ): M(\mathbb{C}^{d})^{\otimes n} \to M(\mathbb{C}^{d})^{\otimes n}
$$
is the linear extension of the map
$$
\mathbb{E}(K_{1}\otimes \dots\otimes K_{n} |A) := 
\prod_{j\in A^{c}} {\rm Tr}(\rho K_{j}) \cdot \prod_{i\in A} K_{i}^{(i)} 
$$
In particular we define
$$
\mathbb{E}(K|\emptyset):= \mathbb{E}(K)= {\rm Tr}(\rho^{\otimes n} K). 
$$
\end{definition}


\begin{lemma}
Let $A$ be a subset of $\{1,\dots ,n\}$. Let $L^{2}(\rho^{\otimes A})$ 
be the subspace of $L^{2}(\rho^{\otimes n})$ consisting of operators acting only on 
tensors in $A$
$$
L^{2}(\rho^{\otimes A}) = {\rm Lin}\left\{\prod_{i\in A} K_{i}^{(i)} \right\}
$$
and let $Q_{A}$ be the corresponding orthogonal projection.

Then 
$$
\mathbb{E}(K|A)= Q_{A}(K), \qquad K\in L^{2}(\rho^{\otimes n}).
$$
\end{lemma}


%


\begin{remark}
Note that
\begin{enumerate}
\item
$Q_{\emptyset}$ is the one dimensional projection onto the identity $\mathbf{1}$ seen as 
a vector in $ L^{2}(\rho^{\otimes n})$;
\item
the map $A\mapsto L^{2}(\rho^{\otimes A})$ respects the order by inclusion
$$
L^{2}(\rho^{\otimes B}) \subset L^{2}(\rho^{\otimes A})  
\qquad {\rm iff} \qquad
B\subset A;
$$
\item
the map $A\mapsto L^{2}(\rho^{\otimes A})$ respects intersections
$$
L^{2}(\rho^{\otimes B}) \cap L^{2}(\rho^{\otimes A}) =
L^{2}( \rho^{\otimes B\cap A}) , \qquad {\rm or}\qquad
Q_{A}\cdot Q_{B}= Q_{A\cap B}.
$$
\end{enumerate}
\end{remark}

By the same construction as in the classical case we can define for each $A$ a subspace 
$\mathcal{H}_{A}$ of $L^{2}(\rho^{\otimes A})$ consisting of vectors which are orthogonal to all spaces $\mathcal{H}_{B}$ with $B\neq A$.

\begin{definition}\label{def.HA}
Let $A$ be a subset of $\{1,\dots , n\}$. We denote by $\mathcal{H}_ {A}$ the subspace of $L^{2}(\rho^{\otimes A})$ consisting of operators which satisfy 
$$
\mathbb{E}(K|B)=0, \qquad \mathrm{for~ all} ~B: |B|<|A|.
$$
\end{definition}
It can be ckecked directly that the first such projections are
\begin{eqnarray*}
P_{\emptyset}(H) &=& \mathbb{E} (H) \mathbf{1}\\
P_{\{i\}}(H)&=& \mathbb{E}(H| \{i\}) -   \mathbb{E}(H) \mathbf{1}\\
P_{\{i,j\}} (H)&=&  \mathbb{E}(H| \{i,j\})-  \mathbb{E}(H| \{i\})- \mathbb{E}(H| \{j\})+ 
 \mathbb{E}( H) \mathbf{1}
\end{eqnarray*}

\begin{lemma}\label{orthogonality.Pa} 
Let $P_{A}$ be the orthogonal projection onto $\mathcal{H}_{A}$. Then
$$
P_{A} \cdot P_{B}=0 ,\qquad A\neq B,
$$
i.e. the spaces $\{\mathcal{H}_{A}\}$ are mutually orthogonal. 
\end{lemma}

\begin{proof}
Suppose that $B$ does not contain $A$. Let $H\in \mathcal{H}_{A}$, then  
$$
\mathbb{E}(H | B) = \mathbb{E}(H |A \cap B)=0
$$
hence $Q_{B} (H)=0$. Since $P_{B}\leq Q_{B}$ we obtain $P_{B}(H)=0$. 
\end{proof}

\begin{theorem}[Hoeffding decomposition]
The sets of projections $\{P_{A}\} $ and $\{Q_{A}\}$ are related by
$$
P_{A}= \sum_{B\subset A} (-1)^{|A|-|B|} Q_{B} 
\qquad
{\rm and}
\qquad
Q_{A} = \sum_{B\subset A} P_{B}.
$$
\end{theorem}

\begin{proof}

Let us denote $\tilde{P}_{A}:=\sum_{B\subset A} (-1)^{|A|-|B|} Q_{B} $. 
Let $C\subsetneqq A$.
By multiplying the right side of the first identity by $Q_{C}$ we get
\begin{eqnarray*}
\tilde{P}_{A}Q_{C}&=&\sum_{B\subset A} (-1)^{|A|-|B|} Q_{B\cap C}=
\sum_{D\subset C} Q_{D} \sum_{D\subset B\subset A} (-1)^{|A|-|B|}  \\
&=&\sum_{D\subset C} Q_{D} \sum_{j=0}^{|A|-|C|} (-1)^{|A|- |D|-j} \binom{|A|-|C|}{j}=0
\end{eqnarray*}
since the interior sum is equal to zero by the binomial formula.This means that 
$
{\rm Ran}\left(\tilde{P}_{A}\right) \subset \mathcal{H}_{A}.
$ 
Now it suffices to show that for any $H\in L^{2}(\rho^{\otimes n})$  and 
$K_{A} \in \mathcal{H}_{A}$ we have
$$
\langle H- \tilde{P}_{A}(H) , K_{A}  \rangle_{\rho^{\otimes n}} =
\langle P_{A}(H)- \tilde{P}_{A}(H) , K_{A}  \rangle_{\rho^{\otimes n}}  =0.
$$
By the definition of  $\tilde{P}_{A}$ we have
$$
\langle H- \tilde{P}_{A}(H) , K_{A}  \rangle_{\rho^{\otimes n}}=
\langle H- Q_{A}(H) , K_{A}  \rangle_{\rho^{\otimes n}} -
\sum_{B\subsetneqq A} (-1)^{|A|-|B|} 
\langle Q_{B}(H), Q_{B}(K_{A})\rangle_{\rho^{\otimes n}}=0
$$
since $(1-Q_{A}) P_{A}=0$ and $Q_{B}P_{A}=0$ for $B\subsetneqq A$.

Using the first identity, we prove the second identity by induction on $r=|A|$. For 
$A=\emptyset$ we have $P_{A}= Q_{A}$ by definition. Suppose that the assertion is true for all sets of size $0,\dots r-1$, and let $A$ be a set of size $r$. 
Then from the first identity, and using the induction hypothesis we get
\begin{eqnarray*}
Q_{A}
&=& 
P_{A} - \sum_{B\subsetneqq A} (-1)^{|A| - |B| }\sum_{C\subset B} P_{C}\\
&=& 
P_{A}- \sum_{C\subsetneqq A} P_{C} \sum_{ B}  (-1)^{|A| - |B|}\\
&=&
\sum_{C \subset A} P_{C}.
\end{eqnarray*}
The interior sum in the second line is over $ \{ B: B\subsetneqq A ~{\rm and}~ B \supset C\}$ 
and can be easily computed using the binomial formula.

\end{proof}

\section{Quantum U-statistics}\label{sec.quantum.u.stat}

In this section we introduce the notion of quantum $U$-statistic and prove the main result formulated in the convergence Theorem \ref{th.conv.moments.ustat}.

 A quantum U-statistic is defined in analogy to its classical counterpart. We replace the sample $X_{1},\dots, X_{n}$ by $n$ independent, identical quantum systems, each of them prepared in a state $\rho\in M(\mathbb{C}^{d})$. The joint state $\rho^{\otimes n}$ is invariant under the action of the symmetric group 
$S(n)$ 
$$
\rho^{\otimes n} = 
\pi_{d}(\tau)\, \rho^{\otimes n}\, \pi_{d}(\tau)^{*},\qquad \tau\in S(n),
$$ 
where $\pi_{d}(\tau)$ is the natural unitary representation of $S(n)$ on $(\mathbb{C}^{d})^{\otimes n}$ 
$$
\pi_{d}(\tau): \psi_{1} \otimes   \dots \otimes  \psi_{n} \mapsto 
\psi_{\tau^{-1}(1)} \otimes \dots \otimes  \psi_{\tau^{-1}(n)}.
$$
Additionally, $(\mathbb{C}^{d})^{\otimes n}$ carries the tensor representation of the special unitary group $SU(d)$ 
$$
\pi_{n} (u) :\psi_{1} \otimes   \dots \otimes  \psi_{n} \mapsto u \psi_{1} \otimes   \dots \otimes  u \psi_{n}  , \qquad u\in SU(d).
$$
These two representations are related by the following duality \cite{Goodman&Wallach}. 
\begin{theorem}[Weyl duality]\label{th.sud.sn}
Let $\pi_{n}$ and $\pi_{d}$ be the representations of $SU(d)$ and respectively 
$S(n)$ on $(\mathbb{C}^{d})^{\otimes n}$. Then the representation space decomposes into a direct sum of tensor products of irreducible representations of $SU(d)$ and 
$S(n)$ indexed by Young diagrams $\lambda$ with $d$ lines and $n$ boxes:
\begin{eqnarray*}
(\mathbb{C}^{d})^{\otimes n} &\cong& 
\bigoplus_{\lambda}
\mathcal{H}_{\lambda}\otimes \mathcal{K}_{\lambda} ,\\
\pi_{n} \equiv 
\bigoplus_{\lambda} 
\pi_{\lambda} \otimes \mathbf{1}_{\mathcal{K}_{\lambda}}, &&
\pi_{d} \equiv
\bigoplus_{\lambda} 
\mathbf{1}_{\mathcal{H}_{\lambda}} \otimes \pi_{\lambda}.
\end{eqnarray*}
\end{theorem}
We will need this result only in as much as saying that the matrix algebras 
$\mathcal{A}_{n}:= {\rm Alg}(\pi_{n} (SU(d)))$ and $\mathcal{B}_{n}:= {\rm Alg}(\pi_{d} (S(n)))$ generated by the two group representations are each other's commutant: 
$
\mathcal{A}_{n}=\mathcal{B}_{n}^{\prime}.
$
Recall that the commutant of a set of operators $\mathcal{M}\subset \mathcal{B}(\mathcal{H})$ is defined as
$$
\mathcal{M}^{\prime} = \{N \in  \mathcal{B}(\mathcal{H}) : [N, M]=0, {\rm any~} M\in \mathcal{M}\}.
$$
The following algebraic lemma will be useful in analysing the asymptotic distribution of $U$-statistics.
\begin{lemma}
The algebra $\mathcal{A}_{n}$ of operators which are invariant under permuatations is generated by the collective observables 
$
\bar{A}:=\sum_{i =1}^{n} A^{(i)}
$
with $A\in M(\mathbb{C}^{d})$. Moreover, any symmetric (permutation invariant) operator can be written as a `symmetric polynomial'  in collective observables, i.e. a linear combination of operators of the form
\begin{equation}\label{eq.def.symmetric.poly}
\mathcal{S}(\bar{A}_{1}, \dots, \bar{A}_{k}):= \frac{1}{k!} \sum_{\tau\in S(k)} \bar{A}_{\tau(1)}\dots \bar{A}_{\tau(k)}.
\end{equation}

\end{lemma}

\begin{proof}

Let $\mathcal{A}:= {\rm Alg}(\pi_{n} (SU(d)))$ and let 
$\mathcal{B}:= {\rm Alg}(\pi_{d} S(n))$ be the matrix algebras generated by the two group representations such that
$
\mathcal{B}^{\prime}=\mathcal{A}.
$
Any element of $SU(d)$ can be written as $u=\exp(iA)$ with 
$A\in M(\mathbb{C}^{d})$ selfadjoint. Then $\pi_{d}(u)= \exp(i\bar{A})$ and by differentiation one can conclude that $\{\bar{A}: A\in M (\mathbb{C}^{d})\}$ generate the algebra $\mathcal{A}$.

\vspace{2mm}

 The second statement can be proved by induction with $k$. Indeed the map 
$A\to \bar{A}$ is a Lie algebra morphism so that 
$
[\bar{A}, \bar{B}]= \bar{C}$ where $C:= [A,B]$. For $k=2$
$$
S(\bar{A}_{1},\bar{A}_{2})= \frac{1}{2}(\bar{A}_{1}\bar{A}_{2} + \bar{A}_{2}\bar{A}_{1})= 
\bar{A}_{1}\bar{A}_{2} + \frac{1}{2}[\bar{A}_{2}, \bar{A}_{1}]= 
\bar{A}_{1}\bar{A}_{2} + \frac{1}{2}\bar{A}, \qquad
A:= [A_{2}, A_{1}]
$$
hence $\bar{A}_{1}\bar{A}_{2}$ is a linear combination of symmetric products. Similarly, but using commutators repeatedly we get
$$
S(\bar{A}_{1},\dots ,\bar{A}_{k+1}) = \bar{A}_{1}\dots \bar{A}_{k+1} + R(k) 
$$
where $R(k)$ is a linear combination of symmetric monomials of order at most $k$.

\end{proof}

 In analogy to the classical case, a selfadjoint operator $K$ on $(\mathbb{C}^{d})^{\otimes r}$ satisfying
$
K=  \pi_{r}(\tau)\, K \,\pi_{r}(\tau)
$
for all $\tau\in S(r)$ will be called symmetric kernel of order $r$. By tensoring with the identity operator on $(\mathbb{C}^{d})^{\otimes n-r}$, the kernel $K$ will be identified with its ampliation to $(\mathbb{C}^{d})^{\otimes n}$, acting as 
$
K\otimes \mathbf{1}\otimes \dots \otimes \mathbf{1}. 
$ 
Sometimes it is useful to express $K$ as a sum of tensor products
$$
K= \sum_{a=1}^{p} K_{a,1} \otimes K_{a,2} \dots \otimes K_{a,r} 
$$
with $K_{a,i}$ selfadjoint operators. Then (the ampliation of) $K$ can be written as 
$$
K=  \sum_{a=1}^{p} \, \prod_{l=1}^{m} K_{a,l}^{(l)}
$$
where we used the notation \eqref{eq.action.position.i}. Similarly, for any unordered subset 
$\beta =\{\beta_{1},\dots,\beta_{r}\}$ of $\{1,\dots, n\}$ we define
$K^{(\beta)}$ by
$$
K^{(\beta)}=  \sum_{a=1}^{p} \, \prod_{l=1}^{r} K_{a,l}^{(\beta_{l})}
$$ 
which is the version of $K$ acting on the tensors on positions 
$ \beta_{1},\dots, \beta_{r}$.

\begin{definition}
Let $K$ be a symmetric kernel of order $r$. 
The $U$-statistic with kernel $K$ is defined by 
\begin{equation}\label{eq.def.un.quantum}
U_{n}:= 
\binom{n}{r}^{-1} \sum_{\beta}  K^{(\beta)}
\end{equation}
where the sum is taken over all unordered subsets $\beta$ of integers from 
$\{1,\dots,n\}$. 
\end{definition}
As in the classical case, the $U$-statistic can be seen as an unbiased estimator of the parameter
$$
\theta:= {\rm Tr}(K\rho^{\otimes r})=  {\rm Tr} (\rho^{\otimes n} U_{n})
$$
which depends on the unknown state $\rho$. The difference with the classical case is that the terms $K^{(\beta)}$ 
of the sum \eqref{eq.def.un.quantum} may not commute with each other, in which case the U-statistic cannot be reduced to a classical one by simultaneously diagonalising all 
$K^{(\beta)}$.

We will now use the Hoeffding decomposition discussed in section \ref{sec.q.hoeffding} to analyse the asymptotic behaviour of $U_{n}$. Since the variance is 
a Hilbert space rather than algebraic property it can be computed in precisely the same way as in the classical case. 
\begin{lemma}\label{lemma.dec.Un.Uc}
Let $U_{n}$ be a $U$-statistic with kernel $K$ of order $r$. 
Then $U_{n}$ has the following decomposition
\begin{equation}\label{eq.dec.Un.Uc}
U_{n} = \sum_{l=1}^{r} \binom{r}{l} U_{n,l}
\end{equation}
where $U_{n,l}$ is the $U$-statistic of order $l$ with kernel $K_{l}$ given by the projection
$
K_{l}:= P_{\{1,\dots, l\}} K.
$
Moreover, the terms $U_{n,l}$ are mutually orthogonal in $L^{2}(\rho^{\otimes n})$ and 
$$
{\rm Var}_{\rho}(U_{n}) 
= 
\sum_{l=1}^{r} \binom{r}{l}^{2} \binom{n}{l}^{-1} \mathbb{E}_{\rho}(K_{l}^{2}).
$$
If $K_{1}=\dots = K_{c-1}=0$ and $K_{c}\neq 0$ then ${\rm Var}_{\rho}(U_{n})$ is of 
order $O(n^{-c})$.
\end{lemma}

\begin{proof}
Using the Hoeffding decomposition we can write
$$
K = \sum P_{A} (K) := \sum_{A}K_{A} = \sum_{l=0}^{r} \sum_{|A| = l} K_{A}
$$
where the first sum is taken over all subsets $A$ of $\{1,\dots n\}$. By symmetry all $K_{A}$ with $|A|=l$ have the same $U$-statistic $U_{n,l}$. The coefficient $\binom{r}{l}$ follows from a simple counting argument.

The proof of the second statement is identical to that of Lemma \ref{lemma.variance.un}.
\end{proof}

The next Lemma deals with the algebraic structure of the $U$-statistic. Roughly speaking a $U$-statistic can be expressed as a polynomial in fluctuation observables with (non-commutative) coefficients which are products of empirical averages. Using this expression one can then easily derive the convergence of moments of $U$-statistics by applying Corollary \ref{cor.1}. The main idea is best conveyed through an example. If  
$$
K= \frac{1}{3!} \sum_{\tau \in S(3)}A_{\tau(1)} \otimes A_{\tau(2)} \otimes A_{\tau(3)},
$$ 
then 
\begin{eqnarray*}
\frac{3!}{n^{3/2}} \binom{n}{3}U_{n} &=& 
\frac{1}{n^{3/2}}\sum_{(\beta_{1},\beta_{2},\beta_{3})} 
A_{1}^{(\beta_{1})}A_{2}^{(\beta_{2})} A_{3}^{(\beta_{3})} \\
&=&\mathcal{S} \left( 
\mathbb{F}(A_{1}), \mathbb{F}(A_{2}),\mathbb{F}(A_{3})\right) -
\mathbb{P}_{n} (A_{1}\circ A_{2}) \circ \mathbb{F}_{n}(A_{3})- \\
&&\mathbb{P}_{n} (A_{1}\circ A_{3}) \circ \mathbb{F}_{n}(A_{2})-
\mathbb{P}_{n} (A_{2}\circ A_{3}) \circ \mathbb{F}_{n}(A_{1})+
\frac{2}{\sqrt{n}} \mathbb{P}_{n}(\mathcal{S}(A_{1}, A_{2}, A_{3})).
\end{eqnarray*}
The sum in the first line runs over all sets of different indices 
$(\beta_{1}, \beta_{2},\beta_{3})$ from $\{1, \dots , n\}$ (position taken into account), and can be obtained by subtracting from the symmetric product of fluctuations, those terms for which either two or three position indices coincide. These terms can again be written in terms of symmetric products between fluctuation operators $\mathbb{F}_{n}(A_{i})$ and empirical averages $\mathbb{P}_{n}(A_{i}\circ A_{j})$, or just $\mathbb{P}_{n}(\mathcal{S}(A_{1}, A_{2}, A_{3}))/\sqrt{n}$.  

\begin{definition}\label{def.symmetric.poly}
A polynomial of order $l$ in the non-commutative variables 
$X_{1},\dots ,X_{t}$ is called symmetric if it is of the form
$$
P(X_{1},\dots, X_{t})= 
\sum_{k=1}^{l}  \sum_{\{j_{1},\dots ,j_{k}\}}c_{j_{1},\dots,j_{k}}
\mathcal{S}(X_{j_{1}}, \dots, X_{j_{k}})
$$ 
where 
$$
\mathcal{S}(X_{j_{1}}, \dots, X_{j_{k}}):=\frac{1}{k!}
\sum_{\tau\in S(k)} X_{j_{\tau(1)}}\dots X_{j_{\tau(k)}},
$$
and the real valued coefficients $c_{j_{1},\dots, j_{k}}$ are symmetric functions of 
the indices, with the interior sum taken over all subsets $\{j_{1},\dots, j_{k}\}$ of 
$\{1,\dots , t\}$.
\end{definition}

\begin{lemma}\label{lemma.polynom.un}
Let $A_{1},\dots, A_{l}\in M(\mathbb{C}^{d})_{sa}$ such that ${\rm Tr}(\rho A_{i})=0$ for all $i=1,\dots ,l$, and let $K$ be the symmetric kernel of order $l$
$$
K:=\frac{1}{l!} \sum_{\tau\in S(l)} A_{\tau(1)}\otimes\dots \otimes A_{\tau(l)}.
$$
Then there exist symmetric polynomials $P_{0},\dots ,P_{l-2}$ of order strictly smaller than 
$l$ such that the corresponding $U$-statistic $U_{n}$ can be expressed as 
\begin{eqnarray}
\frac{l!}{n^{l/2}} \binom{n}{l} U_{n}&=&
\mathcal{S}(\mathbb{F}_{n}(A_{1}) ,\dots ,\mathbb{F}_{n}(A_{l}))\nonumber \\
&+&\sum_{t=0}^{l-2} n^{-t/2} 
P_{t} (\mathbb{F}_{n}(A_{1}),\dots ,\mathbb{F}_{n}(A_{l}) , 
\mathbb{P}_{n} (B_{1}) ,\dots ,\mathbb{P}_{n}(B_{p}))
\label{eq.un.polynomial}
\end{eqnarray}
where $p$ is a natural number depending only on $l$ and 
$B_{1},\dots ,B_{p}\in M(\mathbb{C}^{d})_{sa}$. 
 
\end{lemma}

\begin{proof}

Let us first note that the polynomials have a very simple form for small $l$. If $l=1$ then
$$
U_{n} = \frac{1}{n} \sum_{\beta} A_{1}^{(\beta)} = n^{-1/2} \mathbb{F}_{n}(A_{1}) , \qquad {\rm so}\qquad
 \frac{1}{n^{1/2}} \binom{n}{1}U_{n}= \mathbb{F}_{n}(A_{1}).
$$
If $l=2$ then
$$
\frac{2}{n} \binom{n}{2}
U_{n} = \frac{1}{n}\sum_{(\beta_{1},\beta_{2})} 
A_{1}^{(\beta_{1})} 
A_{2}^{(\beta_{2})}=
\mathbb{F}_{n}(A_{1})\circ \mathbb{F}_{n}(A_{2}) - \mathbb{P}_{n}(A_{1}\circ A_{2}).
$$
The case $l=3$ was discussed above. In general 
\begin{eqnarray*}
\frac{l!}{n^{l/2}} \binom{n}{l}U_{n}
&=&
n^{-l/2}\sum_{(\beta_{1},\dots , \beta_{l})} A_{1}^{(\beta_{1})} \dots A_{l}^{(\beta_{l})}\\
&=&
\mathcal{S}(\mathbb{F}_{n}(A_{1}), \dots, \mathbb{F}_{n}(A_{l}) ) \\
&&- 
n^{-l/2} \sum_{\mathcal{P}} \sum_{(\gamma_{1},\dots, \gamma_{|\mathcal{P}|})} 
\mathcal{S}(A_{i} : i\in L_{1} )^{(\gamma_{1})} \dots  \mathcal{S}(A_{i} : i\in L_{|\mathcal{P}|} )^{(\gamma_{|\mathcal{P}|})}
\end{eqnarray*}
where the first sum of the bottom line runs over all partitions into subsets
$\mathcal{P}= \{L_{1},\dots, L_{|\mathcal{P}|}\}$ of the set $\{1,\dots ,l\}$, and the interior sum is taken over all sets of position indices $(\gamma_{1},\dots, \gamma_{|\mathcal{P}|})$ taken from $\{1,\dots, n\}$, which are different from each other. We have also used the shorthand notation $\mathcal{S}(A_{i} : i\in L_{1} )$ for the symmetric product of those $A_{i}$ with indices in $L_{1}$.

By repeating the procedure a finite number of times we obtain that the left side is equal to a symmetric polynomial in collective observables whose terms look like
$$
\pm n^{-l/2}\mathcal{S}( \overline{B_{1}} ,\dots,  \overline{B_{k}}),
$$
The operators of $B_{1},\dots, B_{k}\in M(\mathbb{C}^{d})_{sa}$ can be written 
as nested symmetric products depending on the hierarchy of successive coarser and coarser partitions of $\{1,\dots ,l \}$ leading to the final partition $\{M_{1}, \dots, M_{k} \}$. 
For our purposes it suffices to describe the operators associated to subsets of one and two elements:
\begin{eqnarray*}
M_{a}= \{ i \} &\Longrightarrow & B_{a}= A_{i},\\
M_{a}= \{i,j \} &\Longrightarrow & B_{a}= A_{i} \circ A_{j}.
\end{eqnarray*} 
We can now distribute the powers of $n$ inside the symmetric product such that we obtain $\mathbb{F}_{n}(A_{i})$ for singletons and $\mathbb{P}_{n}(A_{i}\circ A_{j})$ for two elements subsets. If the partition contains subsets $M_{a}$ of more than three elements  then we use a $n^{-1}$ factor for each of them to obtain 
$\mathbb{P}_{n}(B_{a})$, and we are left with a factor $n^{-t/2}$ with $t>0$.

\end{proof}

Next we show that the (properly normalised) $U$-statistics $U_{n}$ converge in moments to a symmetric polynomial in the canonical variables of the algebra $CCR(L^{2}_{0}(\rho))$. Intuitively this follows from Lemma \ref{lemma.polynom.un} and Corollary \ref{cor.1}, but as in the classical case, formulating the result in a coordinate-free fashion seems to be more difficult. 
Therefore we will decompose the kernel and the 
$U$-statistic with respect to the natural ortho-symplectic basis of 
$L^{2}_{0}(\rho)$ described in section \ref{sec.clt}. 
Recall that $L^{2}(\rho)$ splits into 3 orthogonal subspaces
$$
L^{2}(\rho)=  \mathbb{C}\mathbf{1} \oplus L^{2}_{0}(\rho)= \mathbb{C}\mathbf{1} \oplus \mathcal{H}_{\rho}\oplus 
\mathcal{H}_{\rho}^{\perp},
$$
where $\mathcal{H}_{\rho}$ is the `classical part' consisting of matrices $D$ that commute with $\rho$  and satisfy ${\rm Tr}(\rho D)=0$. The corresponding canonical variables $\mathbb{B}(D)$ form the centre of $CCR(L^{2}_{0}(\rho))$, hence the label `classical'. The `quantum part' 
$\mathcal{H}_{\rho}^{\perp}$ has an ortho-symplectic basis $\{t_{j,k}: 1\leq j\neq k\leq d\}$ such that 
$(Q_{j,k}:= \mathbb{G}(t_{j,k}), P_{j,k}= \mathbb{G}(t_{k,j})) $ are canonical variables of $d(d-1)/2$ commuting quantum oscillators indexed by pairs $(j<k)$. 
For simplicity we denote by $\{F_{0}:= \mathbf{1}, F_{1},\dots ,F_{d^{2}-1}\}$ the basis elements 
\begin{equation}\label{eq.basis.f}
\{\mathbf{1} \} \cup \{d_{i} : i=1,\dots d-1\} \cup \{t_{j,k}: 1\leq j\neq k\leq d\}
\end{equation}
arranged in a particular order. Let 
$$
K:= \sum_{i_{1},\dots, i_{r}} k(i_{1}, \dots , i_{r}) 
F_{i_{1}}\otimes \dots\otimes F_{i_{r}}
$$ 
be the decomposition of $K$ with $k(i_{1}, \dots , i_{r})$ a symmetric function. Then 
\begin{eqnarray*}
K_{l} := 
P_{\{ 1,\dots, l \}}(K) &=& \sum_{i_{1},\dots,i_{l} \geq 1} 
k(i_{1},\dots, i_{l},0,\dots , 0) F_{i_{1}}\otimes \dots \otimes F_{i_{l}}
\otimes \mathbf{1}^{r-l}\\
&=&
\sum_{i_{1},\dots,i_{l} \geq 1} 
k_{l}(i_{1},\dots, i_{l})F_{i_{1}}\otimes \dots \otimes F_{i_{l}}
\otimes \mathbf{1}^{r-l}\\
&=&
\sum_{ \{ i_{1},\dots,i_{l}\} } k_{l}( i_{1},\dots, i_{l} ) F_{i_{1},\dots ,i_{l}}
\otimes \mathbf{1}^{r-l}
\end{eqnarray*}
where the last sum is taken over {\it unordered sets} $\{i_{1}, \dots, i_{l}\}$ with all 
$i_{j}\neq 0$, $F_{i_{1},\dots ,i_{l}}$ is the corresponding symmetric kernel 
and $k_{l}$ is a symmetric function of the indices.

The unordered set of indices $\{i_{1}, \dots, i_{l} \}$ is in one-to-one correspondence with the set of multiplicities ${\bf m}:= \{m_{1}, \dots, m_{d^{2}-1}\} $ where 
$
m_{p} = \left| \{ j: i_{j}= p\} \right|
$ 
and we will further write $F_{\bf m}:= F_{i_{1},\dots ,i_{l}}$ and express $K_{l}$ as
\begin{equation}\label{kernel.multiplicity}
K_{l}= \sum_{{\bf m}\in M_{l}} k_{\bf m} F_{\bf m} \otimes {\bf 1}^{r-l}
\end{equation}
where the sum runs over the set of partitions
$$
M_{l}:= \left\{{\bf m}= \{m_{1}, \dots , m_{d^{2}-1} \} : m_{i}\geq 0, \sum_{i=1}^{d^{2}-1}m_{i}=l\right\}.
$$
By inserting into \eqref{eq.dec.Un.Uc} we get
\begin{equation}\label{eq.un.orthonormal}
U_{n}=  \sum_{l=1}^{r} \binom{r}{l}  U_{n,l}=\sum_{l=1}^{r} \binom{r}{l} 
\sum_{{\bf m}\in M_{l}} k_{\bf m} U_{n,{\bf m}}.
\end{equation}


\begin{theorem}[convergence in moments]\label{th.conv.moments.ustat}
Let $U_{n}$ be the U-statistic with kernel $K$ of order $r$ 
such that $K_{1}=\dots K_{c-1}=0\neq K_{c}$ for some $1\leq c\leq r$, with $K_{l}$  given 
by \eqref{kernel.multiplicity}. Then the following convergence in moments holds
$$
\lim_{n\to \infty}\,{\rm Tr}\left(\rho^{\otimes n}  \left( n^{c/2}( U_{n}-  \mathbb{E}_{\rho}(K)) \right)^{p} \right)=\varphi (U^{p}), \qquad p\in \mathbb{N},
$$
where $\varphi$ is the Gaussian state \eqref{eq.quasifree} and $U$ is the polynomial 
in canonical variables of $CCR(L^{2}_{0}(\rho))$
\begin{equation}\label{eq.th.limit.symm.hermite}
U: =\binom{r}{c}\sum_{{\bf m}\in M_{c}} k_{\bf m} \,
\mathcal{S}\left[ \prod_{i}H_{m_{i}}(\mathbb{G}(F_{i}))\right] .
\end{equation}
\end{theorem}

\begin{remark}
Note that the symmetric product in \eqref{eq.th.limit.symm.hermite} is equal to the 
(usual) product of mutually commuting operators 
$$
\prod_{j<k} \mathcal{S}\left[ H_{m(j,k)}(Q_{j,k}/\sqrt{2} \sigma_{j,k}) 
H_{m(k,j)}(P_{j,k}/\sqrt{2} \sigma_{j,k})\right] 
\prod_{i=1}^{d-1} H_{m(i)} (G_{i})
$$
where each term is associated to an oscillator or a classical normal variable. It is known that a polynomial in canonical variables may not be essentially selfadjoint on the domain consisting of finite particles vectors \cite{Nagel}. This is related to the famous Hamburger moment problem \cite{Simon.moment} and the fact that the distribution of powers $X^{k}$ of Gaussian variables $X$ are not uniquely determined by their moments \cite{Stoyanov}. 

For classical (commutative) $U$-statistics the above convergence holds in the stronger sense of convergence in distribution. In the quantum case this may still be true but it is not clear to us how to identify the limit distribution. However, in 
section \ref{sec.nondegenerate.sec.order} we will show that convergence in distribution holds for non-degenerate statistics and for order two kernels.

\end{remark}
\begin{proof}
With $\theta=\mathbb{E}_{\rho} (K)= {\rm Tr}(\rho^{\otimes r}K)$, let
$$
U_{n}-\theta =  
\sum_{l=c}^{r}\binom{r}{l}U_{n,l}= 
\sum_{l=c}^{r}\binom{r}{l}\sum_{{\bf m}\in M_{l}} k_{\bf m} U_{n,{\bf m}}
$$
be the decomposition of $U_{n}$ into simpler $U$-statistics with kernels $F_{\bf m}$. 

By Lemma  \ref{lemma.polynom.un} the scaled components $U_{n,{\bf m}}$ can be 
expressed as polynomials in variables of the form $\mathbb{F}_{n}(F_{i})$ or 
$\mathbb{P}_{n}(B_{j})$ having the form
$$
\frac{l!}{n^{l/2}}\binom{n}{l}U_{n,{\bf m}} = 
\mathcal{S}\left( \prod_{i} \mathbb{F}_{n}(F_{i})^{m_{i}} \right)
+ \sum_{t=0}^{l-2} n^{-t/2}P_{t,{\bf m}} \,, \qquad {\bf m} \in M_{l}
$$
where we used a shorthand notation for the symmetric product (see definition 
\ref{def.symmetric.poly}) of fluctuations operators $\mathbb{F}_{n}(F_{i})$, each appearing with multiplicity $m_{i}$. 

In particular, if $l>c$ then $n^{c/2} U_{n,l}$ has coefficients of order $n^{-(l-c)/2}$ and its contribution to the moments of $n^{c/2}(U_{n}-\theta)$ converges to zero as $n\to\infty$. By applying the Corrolary \ref{cor.1} of the quantum CLT we get 
$$
{\rm Tr}\left(n^{c/2}(U_{n}-\theta)^{p}\right) \to m(p)
$$ where $m(p)$ is the $p$-th moment (with respect to $\varphi$) of the following 
polynomial 
$$
U^{\infty}_{c}:=\binom{r}{c}
\sum_{{\bf m}\in M_{c}} k_{\bf m}\,
\mathcal{S}\left( \prod_{i} \mathbb{G}(F_{i})^{m_{i}} 
\right)  + \tilde{P}_{0}(\mathbb{G}(F_{i}): i\in \{1,\dots, d^{2}-1\} ).  
$$
Here $\tilde{P}_{0}$ is the polynomial of order strictly less than 
$c$, obtained by replacing the emipirical averages $\mathbb{P}_{n}(B)$ by the average
${\rm Tr}(\rho B)$, and the fluctuations by their canonical variables. 
Although this could be computed using the combinatorial arguments employed in 
Lemma \ref{lemma.polynom.un}, we prefer to give a more transparent derivation based on the following orthogonality argument.

From the construction of the Hoeffding decomposition we know that 
$U_{n,c}$ is orthogonal in $L^{2}(\rho^{\otimes n})$ on all $U$-statistics $U_{n,{\bf m}}$ with $|{\bf m}|<c$. By appropriate scaling and taking the limit $n\to \infty$ we then obtain 
that the limit polynomial $U^{\infty}_{c}$ is orthogonal in $L^{2}(\phi)$ on all
$U^{\infty}_{\bf m}$ with $|{\bf m}|<c$.

We prove by induction with respect to $c$ that 
\begin{equation}\label{eq.induction.hermite}
U^{\infty}_{\bf m} = 
\mathcal{S}\left( \prod_{i} H_{m_{i}} (\mathbb{G}(F_{i})) \right), \qquad \forall \, |{\bf m}|\leq c.
\end{equation}

If $c=1$ then the $U$-statistics are proportional to fluctuation operators and by 
CLT we have that $U^{\infty}_{\bf m}$ are just the canonical variables which together with the identity span the space of polynomials of order $1$. Suppose that the induction hypothesis is true for $c-1$ and let ${\bf m}$ be a multiplicity with 
 $|{\bf m}| =c$. Since $U_{\bf m}^{\infty}$ is orthogonal on all $U_{{\bf m}^{\prime}}^{\infty}$ with $|{\bf m}^{\prime}|<c $, it is orthogonal on all polynomials of orders 
less  than $c$. This follows from the fact that the symmetric products of Hermite polynomials of order less than $c$ span the set $\mathcal{P}_{c-1}$ of all polynomials of order less than $c$.

Let $Z_{i}^{\perp}$ denote the projection onto $\mathcal{P}_{i}^{\perp}$. Then
$$
U_{\bf m}^{\infty} =  Z_{c-1}^{\perp} U_{\bf m}^{\infty} = Z_{c-1}^{\perp} 
\mathcal{S}\left( \prod_{i} \mathbb{G}(F_{i})^{m_{i}} \right),
$$
since $\tilde{P}_{0}\in \mathcal{P}_{c-1}$.
Using Corollary \ref{cor.hermite.orthogonal}, we obtain
$$
U_{\bf m} = 
Z_{c-1}^{\perp} 
\mathcal{S}\left( \prod_{i} \mathbb{G}(F_{i})^{m_{i}} \right)
=
\mathcal{S}\left(  \prod_{i} 
H_{m_{i}} (\mathbb{G}(F_{i}))
\right),
$$ 
which ends the proof of the induction hypothesis.

\end{proof}

\section{Non-degenerate and second order kernels}

\label{sec.nondegenerate.sec.order}

In this section we show that Theorem \ref{th.conv.moments.ustat} can be strengthened in 
some cases of particular interest to applications: non-degenerate kernels, and kernels 
of order 2.  

\begin{corollary}\label{cor.gaussian.limit}
A kernel $K$  of order $r$ is called non-degenerate if $K_{1}\neq 0$. In this case the associated $U$-statistic is asymptotically normal:
$$
\sqrt{n}(U_{n}-\theta)\overset{\mathcal{L}}{\longrightarrow}
N(0,r^{2} \xi_{1} ), \qquad \xi_{1}= {\rm Tr}(\rho K_{1}^{2}).
$$ 
\end{corollary}

\begin{proof}
The convergence in moments to $N(0,r^{2} \xi_{1} )$ follows from Theorem \ref{th.conv.moments.ustat}. Since the limit distribution is uniquely determined by its moments, this implies that $\sqrt{n}(U_{n}-\theta)$ converges in distribution to 
$N(0,r^{2} \xi_{1} )$ (cf. \cite{Billingsley}).
\end{proof}

\begin{theorem}\label{lemma.order2.convergence.distrib}
Let $K:= \theta\mathbf{1}+\sum_{i,j>0}k_{i,j}F_{i}\otimes F_{j}$ be a degenerate kernel of order 2. Then the associated $U$-statistic is given by
$$
U_{n} -\theta=  \frac{1}{n-1} \sum_{i,j>0}k_{i,j}(\mathbb{F}_{n}(F_{i})\circ \mathbb{F}_{n}(F_{j})  -\mathbb{P}_{n}(F_{i}\circ F_{j})),
$$
and 
$$
(n-1)(U_{n}-\theta)\overset{\mathcal{L}}{\longrightarrow}
U:=\sum_{i,j}k_{i,j>0}(\mathbb{G}(F_{i})\circ \mathbb{G}(F_{j})  -\delta_{i,j} \mathbf{1}).
$$
where the the operator $U$ is selfadjoint.
\end{theorem}

\begin{proof}

The form of $U_{n}$ can be derived as in Lemma \ref{lemma.polynom.un} and the convergence in moments can be deduced from Theorem \ref{th.conv.moments.ustat}.

It is well known that in the case of a CCR algebra over a finite dimensional space, any quadratic form in canonical variables is selfadjoint, and its distribution is uniquely determined by its moments. 
We give here a sketch of the proof and refer to \cite{Reed.Simon.2} for more details. By Nelson's analytic vector Theorem, the selfadjointness follows from the fact that finite particle vectors are analytic for $U$, i.e. satisfy 
$\sum_{p}\| U^{p} \psi\| t^{p}/p!< \infty$ for some $t>0$. 
This can be shown by using the inequality
$$
\|A^{\sharp}_{i_{1}} \dots A^{\sharp}_{i_{p}}\psi_{k}\|\leq \sqrt{(p+k)!}
$$
where $A^{\sharp}_{i} $ represents a creation or annihilation operator and 
$\psi_{k}$ is a vector of $k$ excitations. 

To prove uniqueness of the distribution it is enough to show that the moments
$m_{p}:= \phi(U^{p})$ satisfy $|m_{p}|< CD^{p}p!$. The latter follows from the above inequality with $\psi_{k}$ being the vacuum vector in an appropriate representation of the CCR algebra.  

\end{proof}

%



In the next two examples we consider order 2 kernels on $\mathbb{C}^{2}$ with reference state given by the density matrix 
\begin{displaymath}
\rho=
\left(
\begin{array}{cc}
\lambda & 0\\
0 & 1-\lambda
\end{array}
\right), \qquad \lambda >1/2,
\end{displaymath}
so that $\{\mathbf{1}, \sigma_{x}, \sigma_{y} , \sigma_{z} - (2\lambda-1)\mathbf{1} \}$ 
forms an orthogonal basis of $L^{2}(\rho)$, 
where $\sigma_{x}, \sigma_{y}, \sigma_{z}$ are the Pauli matrices
\begin{displaymath}
\sigma_{x}=
\left(
\begin{array}{cc}
0 & 1\\
1 & 0
\end{array}
\right),
\qquad
\sigma_{y}=
\left(
\begin{array}{cc}
0 & -i\\
i & 0
\end{array}
\right),
\qquad
\sigma_{z}=
\left(
\begin{array}{cc}
1 & 0\\
0 & -1
\end{array}
\right).
\end{displaymath}

\begin{ex}
{\rm
Let $K$ be the symmetric kernel of order 
2 given by 
$K= \frac{1}{2}(\sigma_{x} \otimes \sigma_{y}+ \sigma_{y}\otimes \sigma_{x}) $.
The associated $U$-statistic is
\begin{eqnarray*}
\frac{2}{n} \binom{n}{2} U_{n}
&=&
\frac{1}{n} \sum_{\{i\neq j\}} \sigma_{x}^{(i)} \sigma_{y}^{(j)}  
= 
\mathbb{F}_{n}(\sigma_{x} ) \circ \mathbb{F}_{n}(\sigma_{y}) 
 -\mathbb{P}_{n}(\sigma_{x} \circ\sigma_{y}) =
 \mathbb{F}_{n}(\sigma_{x} ) \circ \mathbb{F}_{n}(\sigma_{y}), 
\end{eqnarray*}
and  by Theorem \ref{lemma.order2.convergence.distrib} it converges in distribution to the law of $U:=2(2\lambda-1)Q\circ P$ with respect to the thermal equilibrium state $\phi_{\sigma}$ with $\sigma^{2}= (2(2\lambda -1))^{-1}$.
}
\end{ex}

\begin{ex}
{\rm
Let $H$ be the symmetric kernel of order 2 given by
$
H=\sigma_{x} \otimes \sigma_{x} +\sigma_{y}\otimes \sigma_{y},
$
The associated $U$-statistic is
\begin{eqnarray*}
\frac{2}{n} \binom{n}{2} U_{n}
&=&
\frac{1}{n} \sum_{\{i\neq j\}} (\sigma_{x}^{(i)} \sigma_{x}^{(j)}   + 
\sigma_{y}^{(i)} \sigma_{y}^{(j)} )
= 
\mathbb{F}_{n}(\sigma_{x} )^{2} + \mathbb{F}_{n}(\sigma_{y})^{2}
-\mathbb{P}_{n}(\sigma_{x}^{2})-  \mathbb{P}_{n}(\sigma_{y}^{2})\\
&=&
 \mathbb{F}_{n}(\sigma_{x} )^{2} + \mathbb{F}_{n}(\sigma_{y})^{2} - 2\mathbf{1}
\end{eqnarray*}
and converges in distribution to 
$$
2(2\lambda-1)( Q^{2}+P^{2})-2\mathbf{1} = 4(2\lambda-1)( \mathbf{N}- \mathbb{E}(\mathbf{N})\mathbf{1} ).
$$

}
\end{ex}

\section{Applications to quantum statistics}
\label{sec.applications}
In this section we sketch a couple of quantum statistical applications of the limit theorems for degenerate order two kernels. 
A more detailed analysis will be pursued in a separate publication.

\subsection{Testing for a particular state}

Let $\rho\in \mathbb{C}^{d}$ be a fixed diagonal density matrix with distinct eigenvalues $\lambda_{1}>\dots >\lambda_{d}$ (cf.  \eqref{eq.rho.diagonal}). Given n i.i.d. samples from an unknown state $\sigma\in \mathbb{C}^{d}$, we would like to test whether $\sigma$ is equal to $\rho$. Roughly speaking we would like to distinguish the following hypotheses
\begin{eqnarray*}
H_{0} : \sigma =\rho,\qquad
H_{1} :\sigma\neq \rho ,
\end{eqnarray*}
and we ignore here the obvious difficulties related to the fact that the two sets of states are not `separated'.

The test $t_{n}$ is a binary measurement with outcomes in $\{0,1\}$. When the result $i$ is obtained, we accept the hypothesis $H_{i}$. We denote the corresponding  POVM elements by $T_{n}$ and $\mathbf{1}-T_{n}$ such that 
$$
\mathbb{P}(t_{n}=0|\sigma)= {\rm Tr}(\sigma T_{n}) , \qquad
\mathbb{P}(t_{n}=1|\sigma)= {\rm Tr}(\sigma (\mathbf{1}-T_{n})).
$$
Adopting the standard statistical set-up \cite{Lehmann} we define the type I error probability of the test $t_{n}$ by
$$
\alpha_{n}:=\mathbb{P}(t_{n}=1|\sigma=\rho)= {\rm Tr}(\rho (\mathbf{1}-T_{n}) ).
$$
The type II error probability of $t_{n}$ at a point $\sigma\neq \rho$
is
$$
\beta_{n}(\sigma):= \mathbb{P}(t_{n}=0| \sigma)= {\rm Tr}(\sigma T_{n}).
$$ 

Let $0<\alpha<1$  be a fixed `significance level'. We say that $t_{n}$ is a level-$\alpha$ test if  $\alpha_{n}\leq \alpha$. Among such  level-$\alpha$ tests we would like to identify those which have the smallest second type error probability $\beta_{n}(\sigma)$. In general a `uniformly most powerful' test over all $\sigma$ may not exist but sometimes one can find an optimal one in a restricted class of tests such as unbiased, covariant, etc.

Here we propose a test based on a $U$-statistic  with kernel
$$
K= 
\sum_{i=1}^{d} (\lambda_{i}\mathbf{1} - E_{i,i})\otimes (\lambda_{i}\mathbf{1} - E_{i,i})
+
\sum_{i\neq j=1}^{d} (T_{j,k} \otimes T_{j,k}), 
$$
where $T_{j,k}$ are the selfadjoint matrices defined in \eqref{generators_algebra}. Then $K$ is an unbiased estimator of $\| \sigma-\rho\|_{2}^{2}$, i.e.
$$
{\rm Tr}(\sigma^{\otimes 2} K)= \| \sigma-\rho\|_{2}^{2}:=\theta.
$$

From its expression we see that $U_{n}$ is degenerate at the state $\rho$, that is 
$K_{1}=0$. By applying Lemma \ref{lemma.order2.convergence.distrib} we get 
$$
nU_{n} \overset{\mathcal{L}}{\longrightarrow}
U:=\sum_{i=1}^{d} (\mathbb{G}_{i}^{2} - \lambda_{i} (1-\lambda_{i}) \mathbf{1}) 
+ \sum_{j< k} (Q_{j,k}^{2} +P_{j,k}^{2}-2\sigma_{j,k}^{2})/\sigma_{j,k}^{2}
$$
where $\{\mathbb{G}_{i}\}$ are centred Gaussian variables with covariance matrix 
$V_{ij}$ and $\{ Q_{j,k}, P_{j,k}\}$ are canonical coordinates of $d(d-1)/2$ quantum oscillators prepared independently in thermal states $\phi_{j,k}$ with variance 
\eqref{eq.variance.thermal}. Let $I(\alpha)= [a,b]$ be a $\alpha$-significance interval for $U$, i.e. $\mathbb{P}(U\notin [a,b])=\alpha$. 

We devise the following test: measure $U_{n}$ to obtain random result $u_{n}$ 
and then accept $H_{0}$ if $n u_{n}\in I(\alpha)$ and $H_{1}$ if $n u_{n}$ is in the `critical region $I(\alpha)^{c}$. The corresponding POVM element $T_{n}$ is the spectral projection of $nU_{n}$ corresponding to the interval $I(\alpha)$. By construction this is an asymptotically level $\alpha$ test, and its type II error at a fixed $\sigma\neq \rho$ can be shown to go exponentially fast to zero by considering the asymptotic distribution of 
$U_{n}$ under $\sigma^{\otimes n}$. In this case the statistic is non-degenerate and we have 
$$
\sqrt{n}(U_{n} - \|\sigma-\rho\|^{2})\overset{\mathcal{L}}{\longrightarrow}
N(0,V)
$$
for some variance $V=V(\rho,\sigma)$. In other words 
$nU_{n}\approx n\|\sigma-\rho\|^{2}+\sqrt{n}N(0,V)$ and 
$$\beta_{n}=\mathbb{P}_{\sigma}( n u_{n}\in I(\alpha))=
O(e^{-n  \|\sigma-\rho\|^{4} /2V})
$$
which means that the test is asymptotically consistent. A more detailed analysis of this test and its `local limiting power' \cite{vanderVaart} (when the alternative 
$\sigma$ converges to the null hypothesis $\rho$ as $n\to\infty$), will be pursued in a separate publication.  

\subsection{Homogeneity test}

We are given $n$ independent systems prepared in a state 
$\sigma_{1}\in M(\mathbb{C}^{d})$ and another $n$-tuple of independent systems with state $\sigma_{2}\in M(\mathbb{C}^{d})$, independent of the first ones. We would like to decide whether $\sigma_{1}=\sigma_{2}$, i.e.
$$
H_{0}:= \{ \sigma_{1}=\sigma_{2}\} \qquad H_{1} := \{ \sigma_{1}\neq \sigma_{2}\}
$$
 
The test is based on a $U$-statistic which gives an unbiased estimator of the distance 
$\|\sigma_{1}-\sigma_{2}\|^{2}$. Let $K$ be the second order kernel on 
$\mathcal{H}=\mathbb{C}^{d}\otimes \mathbb{C}^{d}$
$$
K:=\sum_{i} (E_{i,i}\otimes \mathbf{1} -\mathbf{1} \otimes E_{i,i})\otimes (E_{i,i}\otimes \mathbf{1} -\mathbf{1} \otimes E_{i,i})
+
 \sum_{j\neq k}  (T_{j,k}\otimes \mathbf{1} -\mathbf{1} \otimes T_{j,k})  \otimes 
(T_{j,k}\otimes \mathbf{1} -\mathbf{1} \otimes T_{j,k}).
$$ 
Then 
$$
{\rm Tr}(\sigma_{1}^{\otimes 2}\otimes \sigma_{2}^{\otimes 2} K)= \|\rho-\sigma\|_{2}^{2},
$$
and as in the previous example the $U$-statistic can be used as a test with critical region chosen according to the  distribution of the limit quadratic form $U$. 

%
%

\subsection{Quantum metrology}
The capacity to prepare and control quantum systems opens new possibilities for precision metrology with applications ranging from time and frequency standards 
to magnetometry and gravitational waves detection. Beating the standard quantum limit typically requires either special quantum states of many body systems (e.g. squeezed, or noon states) or special interactions generating the quantum dynamics 
\cite{Boixo,Boixo2,Roy}. We discuss the latter case here in connection with the theory of quantum $U$-statistics developed in this paper. 
Let $K$ be a kernel of order $r$ and let
$$
H_{n}=\sum_{\beta} K^{(\beta)}= {n\choose r}U_{n}
$$
be an $r$-body interaction hamiltonian on $n$ systems, generating a unitary evolution $V_{t}^{\gamma,n} :=\exp( it \gamma H_{n})$ where the coupling constant $\gamma$ is an unknown parameter that we want to estimate. To do this, we let an initial state $\psi_{0}$ evolve for some time and then perform an appropriate measurement on the final state  $\psi^{\gamma,n}_{t}:= 
V_{t}^{\gamma,n} \psi_{0}$ to obtain an estimator $\hat{\gamma}_{n}$. In the special case when $K=h\otimes \dots \otimes h$ with $h$ a fixed selfadjont operator, the hamiltonian can be diagonalised explicitly and an argument based  on the quantum Cram\'{e}r-Rao inequality \cite{Boixo} suggests that the mean square error of $\hat{\gamma}_{n}$ may attain the rate $n^{-2r}$, for special `noon-type' initial states. However, since these states are hard to engineer, it is important to know whether similar rates can be achieved with product states 
 $\psi_{0}=\phi\otimes\dots \otimes \phi$. 
 This is discussed in \cite{Boixo2} where it is shown that the slightly worse rate 
 $n^{-2r+1}$ is achievable for some initial states, and the optimal estimation procedure consists of a sequence of finer and finer estimates, with separate measurements performed at each stage. 
 
Without entering into the details, we suggest that the convergence theory of quantum $U$-statistics offers a general set-up for extending these results to generic $r$-body hamiltonians, mixed initial states, and for obtaining the limit distribution of the estimators. 
Indeed, suppose that $\gamma$ is of order $n^{-r+1/2}$ so that 
$\gamma=n^{-r+1/2} g$, and supposed that 
${\rm Tr} (\rho_{0}^{\otimes r} K)=0$ where $\rho_{0}= |\phi\rangle\langle \phi|$ 
is the initial state of one constituent. Then by Corollary \ref{cor.gaussian.limit} (applied to a pure state) we have
\begin{align*}
\lim_{n\to\infty}\langle \psi_{t}^{\gamma_{2},n}  |\psi_{t}^{\gamma_{1},n}\rangle 
&= \lim_{n\to\infty}
{\rm Tr}\left(\rho_{0}^{\otimes n} \exp( it (g_{1}-g_{2}) n^{-r+1/2} H_{n} ) \right)\\
&= \mathbb{E} \left[ \exp( it(g_{1}-g_{2}) U/r! )\right] =  
\exp(- t^{2}(g_{1}-g_{2})^{2}  \xi_{1} / 2 ((r-1)!)^{2} )=\\
&= \langle \alpha_{1}|\alpha_{2}\rangle
\end{align*}
where $U\sim N(0,r^{2}\xi_{1}),\xi_{1}={\rm Tr}(\rho_{1}K_{1}^{2}) $ and 
$|\alpha_{i}\rangle$ are coherent states with displacement 
$\alpha_{i}= tg_{i} \sqrt{\xi_{1}} /(r-1)!$. This can be interpreted as saying that in 
a  neighbourhood of parameters $\gamma$ of size $n^{-r+1/2}$, 
the output states $\phi_{t}^{\gamma,n}$ can be approximated by coherent states and the parameter can be estimated by a simple homodyne measurement. However, to do this one needs to have already narrowed the parameter down to the 
desired accuracy range! A way out could be to follow the philosophy of {\it local asymptotic normality} \cite{Guta&Kahn,Guta&Kahn2,Guta&Janssens&Kahn} employed for optimal state estimation, and show that the above convergence holds uniformly over a range of parameters $\gamma$ which is {\it larger} that 
$n^{-r+1/2}$, more precisely $n^{-r+1/2+\epsilon}$ for $\epsilon>0$. Taking advantage of this, one could construct a sequence of finer and finer estimators with progressively narrower Gaussian distributions.

\section{Conclusions and outlook}

Symmetric operators on tensor product spaces can be seen as non-commutative functions of collective observables. $U$-statistics are special examples of symmetric operators build from a single kernel $K$ acting on a fixed number of systems. For i.i.d. states, the rescaled $U$-statistics converge in moments to certain polynomials in the canonical variables of the CCR algebra arising from the Central Limit Theorem. In some cases the convergence can be upgraded to convergence in distribution but it remains an open question whether this can be done in general. 
This is closely is related to the problem of finding self-adjoint extensions of polynomials in canonical variables. 

As sketched in the previous section, $U$-statistics may be employed in different statistical problems, notably testing with non-standard hypotheses and precision metrology beyond the Heisenberg limit. We expect that in order to solve these problems, one needs to further develop the theory quantum $U$-statistics. For example, achieving the rate $n^{-r+1/2}$ in parameter estimation requires the control of the {\it rate} of convergence to normality in Corollary \ref{cor.gaussian.limit}. This is a classical 
topic in probability and statistics \cite{vanZwet,Benktus} going back to the 
Berry-Ess\'{e}en bounds for convergence rate in CLT \cite{Berry,Esseen}.


\vspace{3mm}

{\it Acknowledgements.} M.G. would like to thank Gerardo Adesso for fruitful discussions and for pointing out the reference \cite{Boixo2}. 
M.G. was supported by the EPSRC Fellowship EP/E052290/1.


\end{document}